\documentclass[a4paper,cleveref]{article}
\usepackage{fullpage}
\usepackage{xspace,enumitem}
\usepackage[utf8]{inputenc}
\usepackage{thmtools}
\usepackage{thm-restate}
\usepackage{amsfonts, amsmath, amssymb, amsthm}
\usepackage{authblk}
\usepackage{amsmath,amsfonts}
\usepackage{amssymb}
\usepackage[dvipsnames]{xcolor}
\usepackage[cmbtt]{bold-extra}
\usepackage[T1]{fontenc}
\usepackage{comment}
\usepackage{todonotes}
\usepackage[ruled,noline,noend]{algorithm2e}
\usepackage{tabularx}
\usepackage{subcaption}
\usepackage{listings}
\usepackage[hypertexnames=false,colorlinks=true,urlcolor=Blue,citecolor=Green,linkcolor=BrickRed]{hyperref}

\usepackage[colorlinks=true,urlcolor=Blue,citecolor=Green,linkcolor=BrickRed]{hyperref}
\usepackage[capitalise]{cleveref}

\title{Pattern Masking for Dictionary Matching: Theory and Practice}

\author[1]{Panagiotis Charalampopoulos}
\author[2]{Huiping Chen}
\author[3]{Peter Christen}
\author[4]{Grigorios Loukides}
\author[5]{Nadia Pisanti}
\author[6,7]{Solon P.\ Pissis}
\author[8]{Jakub Radoszewski}

\affil[1]{School of Computing and Mathematical Sciences, Birkbeck, University of London, UK\\
\href{mailto:p.charalampopoulos@bbk.ac.uk}{p.charalampopoulos@bbk.ac.uk}}

\affil[2]{School of Computer Science, University of Birmingham, UK\\
    \href{mailto:h.chen.13@bham.ac.uk}{h.chen.13@bham.ac.uk}}

\affil[3]{Australian National University, Canberra, Australia\\
    \href{mailto:peter.christen@anu.edu.au}{peter.christen@anu.edu.au}}

\affil[4]{Department of Informatics, King's College London, UK\\
    \href{mailto:grigorios.loukides@kcl.ac.uk}{grigorios.loukides@kcl.ac.uk}}

\affil[5]{Universit\`a di Pisa, Italy\\
    \href{mailto:nadia.pisanti@unipi.it}{nadia.pisanti@unipi.it}}

 \affil[6]{CWI, Amsterdam, The Netherlands\\
    \href{mailto:solon.pissis@cwi.nl}{solon.pissis@cwi.nl}}
    
 \affil[7]{Vrije Universiteit, Amsterdam, The Netherlands}
 
 \affil[8]{Institute of Informatics, University of Warsaw, Poland\\
    \href{mailto:jrad@mimuw.edu.pl}{jrad@mimuw.edu.pl}}

\newcommand{\cO}{\mathcal{O}}

\def\dd{\mathinner{.\,.}}
\newcommand{\D}{\mathcal{D}}
\newcommand{\M}{\mathcal{M}}

 \newcommand{\PPSMz}{\textsc{PMDM-Size}\xspace}
 \newcommand{\PPSM}{\textsc{PMDM}\xspace}
 \newcommand{\MPPSM}{\textsc{MPMDM}\xspace}
 \newcommand{\MU}{\textsc{MU}\xspace}
 \newcommand{\PPSMw}{\textsc{Heaviest $k$-PMDM}\xspace}
 \newcommand{\PPSMzw}{$k$-\textsc{PMDM}\xspace}
 \newcommand{\kCl}{$k$-\textsc{Clique}\xspace}
 \newcommand{\HS}[1]{\textsc{Heaviest $#1$-Section}\xspace}
 \newcommand{\Greedy}{\textsc{Greedy $\tau$-PMDM}\xspace}
 \newcommand{\HV}{\textsc{HV}\xspace}

\newcommand{\R}{\mathcal{R}}

\newcommand{\TT}{\mathcal{T}}

\newtheorem{theorem}{Theorem}
\newtheorem{remark}{Remark}
\newtheorem{claim}{Claim}
\newtheorem{lemma}{Lemma}

\newtheorem{corollary}{Corollary}

  \newcommand{\defproblem}[3]{
   \vspace{2mm}
 \noindent\fbox{
   \begin{minipage}{0.96\textwidth}
   \textsc{#1}\\
   {\bf{Input:}} #2  \\
   {\bf{Output:}} #3
   \end{minipage}
   }
   \vspace{2mm}
 }

\begin{document}
\date{}
\maketitle
\thispagestyle{empty} 

\begin{abstract}
Data masking is a common technique for sanitizing sensitive data maintained in database systems which is becoming increasingly important in various application areas, such as in record linkage of personal data. This work formalizes the Pattern Masking for Dictionary Matching (\PPSM) problem: given a dictionary $\D$ of $d$ strings, each of length $\ell$, a query string $q$ of length $\ell$, and a positive integer $z$, we are asked to compute a smallest set $K\subseteq\{1,\ldots,\ell\}$, so that if  $q[i]$ is replaced by a wildcard for all $i\in K$, then $q$ matches at least $z$ strings from $\D$. Solving {\PPSM} allows providing data utility guarantees as opposed to existing approaches.

We first show, through a reduction from the well-known $k$-Clique problem, that a decision version of the \PPSM problem is NP-complete, even for binary strings. We thus approach the problem from a more practical perspective. We show a combinatorial $\cO((d\ell)^{|K|/3}+d\ell)$-time and $\cO(d\ell)$-space algorithm for \PPSM for $|K|=\cO(1)$. 
In fact, we show that we cannot hope for a faster combinatorial algorithm, unless the combinatorial $k$-Clique hypothesis fails [Abboud et al., SIAM J.~Comput.~2018; Lincoln et al., SODA 2018]. 
Our combinatorial algorithm, executed with small $|K|$, is the backbone of a greedy heuristic that we propose. Our experiments on real-world and synthetic datasets show that our heuristic finds nearly-optimal solutions in practice and is also very efficient. We also generalize this algorithm for the problem of masking multiple query strings simultaneously so that every string has at least $z$ matches in $\D$.

\PPSM can be viewed as a generalization of the decision version of the dictionary matching with mismatches problem: by querying a \PPSM data structure with string $q$ and $z=1$, one obtains the minimal number of mismatches of $q$ with any string from $\D$. The query time or space of all known data structures for the \emph{more restricted} problem of dictionary matching with at most $k$ mismatches incurs some exponential factor with respect to $k$.
A simple exact algorithm for \PPSM runs in time $\cO(2^\ell d)$. We present a data structure for \PPSM that answers queries over $\D$ in time $\cO(2^{\ell/2}(2^{\ell/2}+\tau)\ell)$ and requires space $\cO(2^{\ell}d^2/\tau^2+2^{\ell/2}d)$, for any parameter $\tau\in[1,d]$.

We complement our results by showing a two-way polynomial-time reduction between \PPSM and the Minimum Union problem [Chlamt\'{a}\v{c} et al., SODA 2017]. This gives a polynomial-time $\cO(d^{1/4+\epsilon})$-approximation algorithm for \PPSM, which is tight under a plausible complexity conjecture.
\end{abstract}

\newpage
\setcounter{page}{1}

\section{Introduction}\label{sec:intro}

Let us start with a true incident to illustrate the essence of the computational problem formalized in this work.
In the Netherlands, water companies bill the non-drinking and drinking water separately. The 6th author of this paper had direct debit for the former but not for the latter. When he tried to set up the direct debit for the latter, he received the following masked message by the company: 

{\small
\begin{lstlisting}
Is this you?
Initial: S.  Name: P****s  E-mail address: s******13@g***l.com 
Bank account number: NL10RABO********11.
\end{lstlisting}}

\noindent The rationale of the data masking is: the client should be able to identify themselves to help the companies \emph{link} the client's profiles, without inferring the  identity of any other client via a \emph{linking  attack}~\cite{DBLP:journals/csur/FungWCY10,DBLP:journals/tkde/Samarati01}, so that clients' privacy is preserved.\footnote{In linking attacks, the adversary is a  data recipient who uses the released data of an individual together with  publicly available data, or with background knowledge, to infer the individual's identity.} Thus, the masked version of the data is required to conceal as few symbols as possible, so that the client can recognize their data, but also to correspond to a sufficient number of other clients, so that it is hard for a successful linking attack to be performed. 

This requirement can be formalized as the Pattern Masking for Dictionary Matching (\PPSM) problem: Given a dictionary $\D$ of $d$ strings, each of length $\ell$, a query string $q$ of length $\ell$, and a positive integer $z$, {\PPSM} asks to compute a smallest set $K\subseteq\{1,\ldots,\ell\}$, so that if $q[i]$, for all $i\in K$, is replaced by a wildcard, $q$ matches at least $z$ strings from $\D$. The {\PPSM} problem applies data masking, a common operation to sanitize personal data maintained in database systems~\cite{san3,san1,san2}. In particular, {\PPSM} lies at the heart of record linkage of databases containing personal data~\cite{Christen2020lsd,sigmodrec18,Kum2014jamia,chi18,podsbestpaper18,vatsalan2017hbbdt}, which is the main application we consider in this work. 

\emph{Record linkage} is the task of identifying records that refer to the same entities across databases, in situations where no entity identifiers are available in these databases~\cite{Christen2012springer,Herzog2007springer,acmsurv20}. This task is of high importance in various application domains featuring personal data, ranging from the health sector and social science research, to national statistics and crime and fraud detection~\cite{Christen2020lsd,Karapiperis2019dpd}. In a typical setting, the task is to link two databases that contain names or other attributes, known collectively as quasi-identifiers (QIDs)~\cite{Vatsalan2014cikm}. The similarity between each pair of records (a record from one of the databases and a record from the other) is calculated with respect to their values in QIDs, and then all compared record pairs are classified into matches (the pair is assumed to refer to the same person), non-matches (the two records in the pair are assumed to refer to different people), and potential matches (no decision about whether the pair is a match or non-match can be made)~\cite{Christen2012springer,Herzog2007springer}.

Unfortunately, potential matches happen quite often~\cite{rlsurv}. A common approach~\cite{chi18,podsbestpaper18} to deal with potential matches is to conduct a manual clerical review, where a domain expert looks at the attribute values in record pairs and then makes a manual match or non-match decision. At the same time, to comply with policies and legislation, one needs to prevent domain experts from inferring the identity of the people represented in the manually assessed record pairs~\cite{chi18}. The challenge is to achieve desired data protection/utility guarantees; i.e.~enabling a domain expert to make good decisions without inferring peoples' identities. 

To address this challenge, we can solve {\PPSM} twice, for a potential match $(q_1,q_2)$. The first time we use as input the query string $q_1$ and a reference dictionary (database) $\mathcal{D}$ containing personal records from a sufficiently large population (typically, much larger than the databases to be linked). The second time, we use as input $q_2$ instead of $q_1$. Since each masked $q$ derived by solving {\PPSM} matches at least $z$ records in $\mathcal{D}$, the domain expert would need to distinguish between at least $z$ individuals in $\mathcal{D}$ to be able to infer the identity of the individual corresponding to the masked string. The underlying assumption is that $\mathcal{D}$ contains one record per individual. Also, some wildcards from one masked string can be superimposed on another to ensure that the expert does not gain more knowledge from combining the two strings, and the resulting strings would still match at least $z$ records in $\mathcal{D}$. Thus, by solving {\PPSM} in this setting, we provide privacy guarantees alike $z$-map~\cite{sweeneythesis}; a variant of the well-studied $z$-anonymity~\cite{10.1145/275487.275508} privacy
model.\footnote{The notation used for such privacy models is generally $k$ instead of $z$, e.g.\ $k$-anonymity~\cite{Samarati1998protecting,sweeneythesis}.} In $z$-map, each record of a dataset must match at least $z$ records in a reference dataset, from which the dataset is derived. In our setting, we consider a  pattern that is not necessarily contained in the reference dataset. Offering such privacy is desirable in real record linkage systems where databases containing personal data are being linked~\cite{Christen2020lsd,Kum2014jamia,vatsalan2017hbbdt}. On the other hand, since each masked $q$ contains the minimum number of wildcards, the domain expert is still able to use the masked $q$ to meaningfully classify a record pair as a match or as a non-match.

Offering such utility is again desirable in record linkage systems~\cite{chi18}. Record linkage is an important application for our techniques, because no existing approach can provide privacy and utility guarantees when releasing linkage results to domain experts~\cite{soups19}. In particular, existing approaches~\cite{Kum2014jamia,soups19} recognize the need to offer privacy by preventing the domain expert from distinguishing between a small number of individuals, but they  provide \emph{no algorithm} for offering such privacy, let alone an algorithm offering utility guarantees as we do.  

A secondary application where PMDM is of importance is \emph{query term dropping},
an information retrieval task that seeks to drop keywords (terms) from a query, so that the remaining keywords retrieve a sufficiently large number of documents. This task is performed by search engines, such as Google~\cite{googlepatent}, and by e-commerce platforms such as e-Bay~\cite{ebaypatent}, to improve users' experience~\cite{ntoulas,sigir17} by making sufficiently many search results available to users. For example, e-Bay applies query term dropping, removing one term, in our test query: 

{\small
\begin{lstlisting}
Query: vacuum database cleaner 
Query results: 0 results found for vacuum database cleaner
              42 results found for vacuum cleaner 
\end{lstlisting}}
\noindent We could perform query term dropping by solving {\PPSM} in a setting where  strings in a dictionary correspond to  document terms and a query string  corresponds to a user's query. Then, we provide the user with the masked query, after removing all wildcards, and with its matching strings from the  dictionary. Two remarks are in order for this application. First, we consider a setting where the keyword order matters. This occurs, for example, when using \emph{phrase search} in Google.\footnote{An indicative example of a query in Google is ``free blue tv'' which yielded $7050$ results, as blue tv is a well-known app, whereas the query ``blue free tv'' yielded only $5$  results.} Second, since the dictionary may contain strings of different length, PMDM should be applied only to the dictionary strings that have the same length as the query string. 

Query term dropping is a relevant  application for our techniques, because existing techniques~\cite{sigir17} do not minimize the number of dropped terms. Rather, they drop keywords randomly, which may unnecessarily shorten the query, or drop keywords based on custom rules, which is not sufficiently generic to deal with all queries. More generally, our techniques can be applied to drop terms from any top-$z$ database query~\cite{topksurv} to ensure there are $z$ results in the query answer.

\paragraph{Related Algorithmic Work.} Let us denote the wildcard symbol by $\star$ and provide a brief overview of works related to \PPSM, the main problem considered in this paper.
\begin{itemize}
\item \emph{Partial Match}: Given a dictionary $\D$ of $d$ strings over an alphabet $\Sigma=\{0,1\}$, each of length $\ell$, and a string $q$ over $\Sigma \sqcup \{\star\}$ of length $\ell$, the problem asks whether $q$ matches any string from $\D$. This is a well-studied problem~\cite{DBLP:conf/stoc/BorodinOR99,DBLP:conf/icalp/CharikarIP02,DBLP:journals/jcss/JayramKKR04,DBLP:journals/jcss/MiltersenNSW98,DBLP:journals/siamcomp/Patrascu11,DBLP:journals/siamcomp/PatrascuT09,DBLP:journals/siamcomp/Rivest76}. Patrascu~\cite{DBLP:journals/siamcomp/Patrascu11} showed that any data structure for the Partial Match problem with cell-probe complexity $t$ must use space $2^{\Omega(\ell/t)}$, assuming the word size is $\cO(d^{1-\epsilon}/t)$, for any constant $\epsilon>0$. The key difference to \PPSM is that the wildcard positions in the query strings are fixed. 
\item \emph{Dictionary Matching with $k$-errors}: A similar line of research to that of Partial  Match has been conducted under the Hamming and edit distances, where, in this case, $k$ is the maximum allowed distance between the query string and a dictionary  string~\cite{DBLP:conf/cpm/Belazzougui09,DBLP:journals/algorithmica/BelazzouguiV16,DBLP:journals/ipl/BrodalV00,DBLP:journals/algorithmica/ChanLSTW10,DBLP:conf/stoc/ColeGL04,DBLP:journals/jal/YaoY97}. 
The structure of Dictionary Matching with $k$-errors is very similar to Partial Match as each wildcard in the query string gives $|\Sigma|$ possibilities for the corresponding symbol in the dictionary strings. On the other hand, in Partial Match the wildcard positions are fixed.

The \PPSM problem is a generalization of the decision version of the Dictionary Matching with $k$-errors problem (under Hamming distance): by querying a data structure for \PPSM with string $q$ and $z=1$, one obtains the minimum number of mismatches of $q$ with any string from $\D$, which suffices to answer the decision version of the Dictionary Matching with $k$-errors problem. The query time or space of all known data structures for Dictionary Matching with $k$-mismatches incurs some exponential factor with respect to $k$.
In~\cite{DBLP:conf/soda/Cohen-AddadFS19}, Cohen{-}Addad et al. showed that, in the pointer machine model, for the reporting version of the problem, one cannot avoid exponential dependency on $k$ either in the space or in the query time. 
In the word-RAM model, Rubinstein showed that, conditional on the Strong Exponential Time Hypothesis~\cite{DBLP:conf/iwpec/CalabroIP09}, any data structure that can be constructed in time polynomial in the total size $||\D||$ of the strings in the dictionary cannot answer queries in time strongly sublinear in $||\D||$. 
\end{itemize}

We next provide a brief overview of other algorithmic works related to \PPSM.

\begin{itemize}
\item \emph{Dictionary Matching with $k$-wildcards}: Given a dictionary $\D$ of total size $N$ over an alphabet $\Sigma$ and a query string $q$ of length $\ell$ over $\Sigma \sqcup \{\star\}$ with up to $k$ wildcards, the problem asks for the set of matches of $q$ in $\D$. This is essentially a parameterized variant of the Partial Match problem. The seminal paper of Cole et al.~\cite{DBLP:conf/stoc/ColeGL04} proposed a data structure occupying $\cO(N\log^k\!N)$ space allowing for $\cO(\ell\!+\!2^k\!\log\log N\!+\!|\textsf{output}|)$-time querying.
This data structure is based on recursively computing a heavy-light decomposition of the suffix tree and copying the subtrees hanging off light children.
Generalizations and slight improvements have been proposed in~\cite{DBLP:journals/mst/BilleGVV14},~\cite{DBLP:journals/tcs/LewensteinMRT14}, and~\cite{DBLP:conf/esa/GawrychowskiLN14}.
In~\cite{DBLP:journals/mst/BilleGVV14} the authors also proposed an alternative data structure that instead of a $\log^k\!\!N$ factor in the space complexity has a multiplicative $|\Sigma|^{k^2}$ factor. Nearly-linear-sized data structures that essentially try all different combinations of letters in the place of wildcards and hence incur a $|\Sigma|^k$ factor in the query time have been proposed in~\cite{DBLP:journals/mst/BilleGVV14,LNV14-stacs}. On the lower bound side, Afshani and Nielsen~\cite{AN16icalp} showed that, in the pointer machine model, essentially any data structure for the problem in scope must have exponential dependency on $k$ in either the space or the query time, explaining the barriers hit by the existing approaches.
\item \emph{Enumerating Motifs with $k$-wildcards}: Given an input string $s$ of length $n$ over an alphabet $\Sigma$ and positive integers $k$ and $z$, this problem asks to enumerate all motifs over $\Sigma \sqcup \{\star\}$ with up to $k$ wildcards that occur at least $z$ times in $s$. As the size of the output is exponential in $k$, the enumeration problem has such a lower bound. 
Several approaches exist for efficient motif enumeration, all aimed at reducing the impact of the output's size: efficient indexing to minimize the output delay~\cite{takeakipolydelay,GM+18}; exploiting a hierarchy of wildcards positions according to the number of occurrences~\cite{tcs-mask}; defining a subset of motifs of fixed-parameter tractable size (in $k$ or $z$) that can generate all the others~\cite{PCGS03,PCGS05}, or defining maximality notions meaning a subset of the motifs that implicitly include all the others~\cite{maxmot09,madmax11}.
\end{itemize}

\paragraph{Our Contributions.} We consider the word-RAM model of computations with $w$-bit machine words, where $w=\Omega(\log (d\ell))$, for stating our results. We make the following contributions:
\begin{enumerate}
    \item (\cref{sec:hardness}) A reduction from the $k$-Clique problem to a decision version of the \PPSM problem, which implies that \PPSM is NP-hard, even for strings over a binary alphabet. The reduction also implies conditional hardness of the \PPSM problem. We also present a generalized reduction from the $(c,k)$-Hyperclique problem~\cite{LWW-soda18}.
    \item (\cref{sec:exact}) A combinatorial $\cO((d\ell)^{k/3}+d\ell)$-time and $\cO(d\ell)$-space algorithm for \PPSM if $k=|K|=\cO(1)$, which is optimal if the combinatorial $k$-Clique hypothesis is true. 
    \item (\cref{sec:many_strings}) We consider a generalized version of \PPSM, referred to as \MPPSM: we are given a collection $\M$ of $m$ query strings (instead of one query string) and we are asked to compute a smallest set $K$ so that, for every $q$ from $\M$, if $q[i]$, for all $i\in K$, is replaced by a wildcard, then $q$ matches at least $z$ strings from dictionary $\D$. We show an $\cO((d\ell)^{k/3}z^{m-1}+d\ell)$-time algorithm for \MPPSM, for $k=|K|=\cO(1)$ and $m=\cO(1)$.
    \item (\cref{sec:DS}) A data structure for \PPSM that answers queries over $\D$ in $\cO(2^{\ell/2}(2^{\ell/2}+\tau)\ell)$ time and requires space $\cO(2^{\ell}d^2/\tau^2+2^{\ell/2}d)$, for any parameter $\tau\in[1,d]$.
    \item (\cref{sec:approx}) A polynomial-time $\cO(d^{1/4+\epsilon})$-approximation algorithm for \PPSM, which we show to be tight under a plausible complexity conjecture.
    \item (\cref{sec:greedy}) A greedy heuristic based on the $\cO((d\ell)^{k/3}+d\ell)$-time algorithm.
    \item (\cref{sec:exp}) An extensive experimental evaluation on real-world and synthetic data demonstrating that our heuristic finds nearly-optimal solutions in practice and is also very efficient. In particular, our heuristic finds optimal or nearly-optimal solutions for \PPSM on a dataset with six million records in less than 3 seconds.
\end{enumerate}
\noindent We conclude this paper with a few open questions in~\cref{sec:finale}. 

This paper is an extended version of a paper that was presented at ISAAC 2021~\cite{charalampopoulos_et_al:LIPIcs.ISAAC.2021.65}. 

\section{Definitions and Notation}\label{sec:prel}
\paragraph{Strings.} An \emph{alphabet} $\Sigma$ is a finite nonempty set whose elements are called \emph{letters}. 
We assume throughout an integer alphabet $\Sigma=[1,|\Sigma|]$. Let $x=x[1]\cdots x[n]$ be a \emph{string} of length $|x|=n$ over $\Sigma$. For two indices $1 \leq i \leq j \leq n$, $x[i\dd j]=x[i]\cdots x[j]$ is the \emph{substring} of $x$ that starts at position $i$ and ends at position $j$ of $x$. By $\varepsilon$ we denote the \emph{empty string} of length $0$. A \emph{prefix} of $x$ is a substring of $x$ of the form $x[1\dd j]$, and a \emph{suffix} of $x$ is a substring of $x$ of the form $x[i\dd n]$. A \emph{dictionary} is a collection of strings. We also consider alphabet $\Sigma_{\star}=\Sigma \sqcup \{\star\}$, where $\star$ is a {\em wildcard} letter that is not in $\Sigma$ and {\em matches} all letters from $\Sigma_{\star}$. Then, given a string $x$ over $\Sigma_{\star}$ and a string $y$ over $\Sigma$ with $|x|=|y|$, we say that $x$ {\em matches} $y$ if and only if $x[i] = y[i] \mbox{\ or\ } x[i] = \star, \mbox{\ for all\ } 1 \leq i \leq |x|$. Given a string $x$ of length $n$ and a set $S\subseteq\{1,\ldots,n\}$, we denote by $x_S=x\otimes S$ the string obtained by first setting $x_S=x$ and then $x_S[i]=\star$, for all $i\in S$. We then say that $x$ is \emph{masked} by $S$.

The main problem considered in this paper is the following.

\defproblem{Pattern Masking for Dictionary Matching (\PPSM)}{A dictionary $\D$ of $d$ strings, each of length $\ell$, a string $q$ of length $\ell$, and a positive integer $z$.}{A smallest set $K\subseteq\{1,\ldots,\ell\}$ such that $q_K\!=\!q\!\otimes\! K$ matches at least $z$ strings from $\D$.}

We refer to the problem of computing only the size $k$ of a smallest set $K$ as \PPSMz.
We also consider the data structure variant of the \PPSM problem in which $\D$ is given for preprocessing, and $q,z$ queries are to be answered on-line.
Throughout, we assume that $k\geq 1$ as the case $k=0$ corresponds to the well-studied dictionary matching problem for which there exists a classic optimal solution~\cite{DBLP:journals/cacm/AhoC75}.
We further assume $z\leq d$; otherwise the \PPSM has trivially no solution. In what follows, we use $N$ to denote $d\ell$.

\paragraph{Tries.} Let $\M$ be a finite set containing $m>0$ strings over $\Sigma$. The \emph{trie} of $\M$, denoted by $\R(\M)$, contains a node for every distinct prefix of a string in $\M$; the root node is $\varepsilon$; the set of leaf nodes is $\M$; and edges are of the form $(u,\alpha,u\alpha)$, where $u$ and $u\alpha$ are nodes and $\alpha\in\Sigma$ is the label. The \emph{compacted trie} of $\M$, denoted by $\TT(\M)$, contains the root, the branching nodes, and the leaf nodes of $\R(\M)$. Each maximal branchless path segment from $\R(\M)$ is replaced by a single edge, and a fragment of a string $M\in \M$ is used to represent the label of this edge in $\cO(1)$ space. The size of $\TT(\M)$ is thus $\cO(m)$. 
The most well-known example of a compacted trie is the suffix tree of a string: the compacted trie of all the suffixes of the string~\cite{DBLP:conf/focs/Weiner73}. 
To access the children of a trie node by the first letter of their edge label in $\cO(1)$ time we use perfect hashing~\cite{DBLP:journals/jacm/FredmanKS84}.
In this case, the claimed complexities hold \emph{with high probability} (w.h.p., for short), that is, with probability at least $1-N^{-c}$ (recall that $N=d\ell$), where $c > 0$ is a constant fixed at construction time. Assuming that the children of every trie node are sorted by the first letters of their edge labels, randomization can be avoided at the expense of a $\log |\Sigma|$ factor incurred by binary searching for the appropriate child.

\section{NP-hardness and Conditional Hardness of \PPSMz}\label{sec:hardness}
We show that the following decision version of \PPSMz is NP-complete. 

\defproblem{\PPSMzw}{A dictionary $\D$ of $d$ strings, each of length $\ell$, a string $q$ of length $\ell$, and positive integers $z\leq d$ and $k\leq \ell$.}{Is there a set $K\subseteq\{1,\ldots,\ell\}$ of size $k$, such that $q_K\!=\!q\!\otimes\! K$ matches at least $z$ strings from $\D$?}

Our reduction is from the well-known NP-complete \kCl problem~\cite{DBLP:books/daglib/p/Karp10}: Given an undirected graph $G$ on $n$ nodes and a positive integer $k$, decide whether $G$ contains a clique of size $k$ (a {\em clique} is a subset of the nodes of $G$ that are pairwise adjacent).

\begin{theorem}\label{thm:red}
  Any instance of the \kCl problem for a graph with $n$ nodes and $m$ edges can be reduced in $\cO(nm)$ time to a \PPSMzw instance with $\ell=n$, $d=m$ and $\Sigma=\{\mathtt{a},\mathtt{b}\}$.
\end{theorem}
\begin{proof}
Let $G=(V,E)$ be an undirected graph on $n=|V|$ nodes numbered $1$ through $n$, in which we are looking for a clique of size $k$. We reduce \kCl to \PPSMzw as follows. 
Consider the alphabet $\{\texttt{a},\texttt{b}\}$.
Set $q=\texttt{a}^n$, and for every edge $(u,v)\in E$ such that $u<v$, add string $\texttt{a}^{u-1}\texttt{b}\texttt{a}^{v-u-1}\texttt{b}\texttt{a}^{n-v}$ to $\D$. Set $z=k(k-1)/2$.
Then $G$ contains a clique of size $k$, if and only if \PPSMzw returns a positive answer.
This can be seen by the fact that cliques of size $k$ in $G$ are in one-to-one correspondence with subsets $K\subseteq \{1,\ldots,n\}$ of size $k$ for which $q_K$ matches $z$ strings from $\D$: the elements of $K$ correspond to the nodes of a clique and the $z$ strings correspond to its edges. 
\PPSMzw is clearly in NP and the result follows. 
\end{proof}

An example of the reduction from \kCl to \PPSMzw is shown in Figure~\ref{fig:redu-kclique}.

\begin{figure}[!t]
\begin{center}
\includegraphics[width=0.99\textwidth]{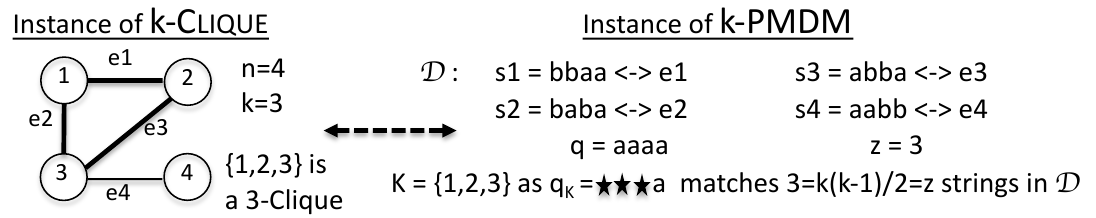}
\end{center}
\caption{An example of the reduction from \kCl to \PPSMzw. The solution for both is $\{1,2,3\}$ as shown. Note that, for $k=4$, the instance of $4$-\textsc{PMDM} would need $z=6$ matches; neither this many matches can be found in $\D$ nor a $4$-clique can be found in the graph.
\label{fig:redu-kclique}}
\end{figure}

\begin{corollary}\label{cor:NP}
\PPSMzw is NP-complete for strings over a binary alphabet.
\end{corollary}

Any algorithm solving \PPSMz can be trivially applied to solve \PPSMzw. 

\begin{corollary}
\PPSMz is NP-hard for strings over a binary alphabet.
\end{corollary}

\begin{remark}
Given an undirected graph $G$, an {\em independent set} is a subset of nodes of $G$ such that no two distinct nodes of the subset are adjacent.
Let us note that the problem of computing a maximum clique in a graph $G$, which is equivalent to that of computing the maximum independent set in the complement of $G$, cannot be $n^{1-\epsilon}$-approximated in polynomial time, for any $\epsilon>0$, unless $\text{P}=\text{NP}$~\cite{clique-hardapprox,DBLP:journals/toc/Zuckerman07}. In Section~\ref{sec:approx}, we show a polynomial-time $\cO(d^{1/4+\epsilon})$-approximation algorithm for \PPSM.
We remark that this algorithm and Theorem~\ref{thm:red} do not contradict the inapproximability results for the maximum clique problem, since our reduction from \kCl to \PPSMzw cannot be adapted to a reduction from maximum clique to \PPSMz.
\end{remark}

Theorem~\ref{thm:red} shows that solving \PPSMzw efficiently even for strings over a binary alphabet would imply a breakthrough for the \kCl problem for which it is known that, in general, no fixed-parameter tractable algorithm with respect to parameter $k$ exists unless the Exponential Time Hypothesis (ETH) fails~\cite{CHKX06,DBLP:journals/jcss/ImpagliazzoP01}. That is, \kCl has no $f(k)n^{o(k)}$ time algorithm, and is thus W[1]-complete (again, under the ETH hypothesis). On the upper bound side, \kCl can be trivially solved in $\cO(n^k)$ time (enumerating all subsets of nodes of size $k$), and this can be improved to $\cO(n^{\omega k/3})$ time for $k$ divisible by $3$ using square matrices multiplication ($\omega$ is the exponent of square matrix multiplication). However, for general $k\geq 3$ and any constant $\epsilon>0$, the \kCl hypothesis states that there is no $\cO(n^{(\omega/3-\epsilon) k})$-time algorithm and no combinatorial $\cO(n^{(1-\epsilon)k})$-time algorithm~\cite{DBLP:journals/siamcomp/AbboudBW18,LWW-soda18,ICM18}.
Thus, conditional on the \kCl hypothesis, and since $d\ell = nm = \cO(n^3)$ and $\ell=n$ (Theorem~\ref{thm:red}), we cannot hope to devise a combinatorial algorithm for \PPSMzw with runtime $\cO((d\ell)^{(1-\epsilon)k/3})$ or $\cO(\ell^{(1-\epsilon)k})$ for any constant $\epsilon>0$.
In~\cref{sec:exact}, we show a combinatorial $\cO(d\ell+ \min \{(d \ell)^{k/3} , \ell^k\})$-time algorithm, for constant $k \geq 3$, for the optimization version of \PPSMzw (seeking to maximize the matches), which can then be trivially applied to solve \PPSMzw in the same time complexity, thus matching the above conditional lower bound.
Additionally, under the \kCl hypothesis, even with the aid of algebraic techniques, one cannot hope for an algorithm for \PPSMzw with runtime $\cO((d\ell)^{(\omega/9-\epsilon) k})$ or $\cO(\ell^{(\omega/3-\epsilon) k})$, for any constant $\epsilon>0$.

In fact, as we show next, by reducing from the \textsc{$(c,k)$-Hyperclique} problem,
which is not known to benefit from fast matrix multiplication~\cite{LWW-soda18}, we obtain stronger conditional lower bounds for some values of $d$ and~$\ell$.

A {\em hypergraph} $H$ is a pair $(V,E)$, where $V$ is the set of nodes of $H$ and $E$ is a set of non-empty subsets of $V$, called {\em hyperedges}. The \textsc{$(c,k)$-Hyperclique} problem is defined as follows: Given a hypergraph $H=(V,E)$ such that all of its hyperedges have size $c$, does there exist a set $S$ of $k>c$ nodes in $V$ so that every subset of $c$ nodes from $S$ is a hyperedge? We call set $S$ a \emph{$(c,k)$-hyperclique} in $H$.
We will reduce \textsc{$(c,k)$-Hyperclique} to \PPSMzw in time $\cO(|V|\cdot|E|)$.

\begin{theorem}\label{thm:red2}
  Any instance of the \textsc{$(c,k)$-Hyperclique} problem for a hypergraph with $n$ nodes and $m$ hyperedges each of size $c$ can be reduced in $\cO(nm)$ time to a \PPSMzw instance with $\ell=n$, $d=m$ and $\Sigma=\{\mathtt{a},\mathtt{b}\}$.
\end{theorem}
\begin{proof}
We reduce \textsc{$(c,k)$-Hyperclique} to \PPSMzw as follows. 
Consider the alphabet $\{\texttt{a},\texttt{b}\}$.
Set $q=\texttt{a}^n$, and for every hyperedge $e_i\in E$, add  the binary string $x_i$ to $\D$ such that $x_i[j]=\texttt{b}$ if and only if $j\in e_i$. Set $z=\binom{k}{c}$.
Then $H$ contains a $(c,k)$-hyperclique if and only if \PPSMzw returns a positive answer.
This can be seen by the fact that $(c,k)$-hypercliques $H$ are in one-to-one correspondence with subsets $K\subseteq \{1,\ldots,n\}$ of size $k$ for which $q_K$ matches $z$ strings from $\D$: the elements of $K$ correspond to the nodes of a $(c,k)$-hyperclique and the $z$ strings correspond to its hyperedges. 
\end{proof}

The \textsc{$(c,k)$-Hyperclique} hypothesis states that there is no $\cO(n^{(1-\epsilon)k})$-time algorithm, for any $k>c>2$ and  
$\epsilon>0$, that solves the \textsc{$(c,k)$-Hyperclique} problem. For a discussion on the plausibility of this hypothesis and for more context, we refer the reader to~\cite[Section 7]{LWW-soda18}. Theorem~\ref{thm:red2} shows that solving \PPSMzw efficiently even for strings over a binary alphabet would imply a breakthrough for the \textsc{$(c,k)$-Hyperclique} problem.
In particular, assuming that the \textsc{$(3,k)$-Hyperclique} hypothesis is true, due to~\cref{thm:red2}, and since $d\ell = nm = \cO(n^4)$, we cannot hope to devise an algorithm for \PPSMzw requiring time $\cO((d \ell)^{(1-\epsilon)k/4})$ or $\cO(\ell^{(1-\epsilon)k})$, for any $k>3$ and $\epsilon >0$.

\section{Exact Algorithms for a Bounded Number $k$ of Wildcards}\label{sec:exact}

We consider the following problem, which we solve by exact algorithms. 
These algorithms  will form the backbone of our effective and efficient heuristic for the \PPSM problem (see~\cref{sec:greedy}). 

\defproblem{\PPSMw}{A dictionary $\D$ of $d$ strings, each of length $\ell$, a string $q$ of length $\ell$, and a positive integer $k\leq \ell$.}{A set $K\subseteq\{1,\ldots,\ell\}$ of size $k$ such that $q_K\!=\!q\!\otimes\! K$ matches the maximum number of strings in $\D$.}

We will show the following result, which we will employ to solve the \PPSM problem.

\begin{theorem}\label{the:kWMDM}
\PPSMw for $k\!=\!\cO(1)$ can be solved in $\cO(N\!+\!\min\{N^{k/3},\ell^k\})$ time, where $N=d\ell$.
\end{theorem}

Recall that a {\em hypergraph} $H$ is a pair $(V,E)$, where $V$ is the set of nodes of $H$ and $E$ is a set of non-empty subsets of $V$, called {\em hyperedges}---in order to simplify terminology we will simply call them edges. Hypergraphs are a generalization of graphs in the sense that an edge can connect more than two  nodes. Recall that the size of an edge is the number of nodes  it contains. The {\em rank} of $H$, denoted by $r(H)$, is the maximum size of an edge of $H$.

We refer to a hypergraph $H[K]=(K,\{e : e \in E, e \subseteq K\})$, where $K$ is a subset of $V$, as a $|K|$-\emph{section}. $H[K]$ is   
the hypergraph induced by $H$ on the nodes of $K$, and it contains all edges of $H$ whose elements are all in $K$. 
A hypergraph is {\em weighted} when each of its edges is associated with a weight. We define {\em the weight} of a weighted hypergraph as the sum of the weights of all of its edges. In what follows, we also refer to weights of nodes for conceptual clarity; this is equivalent to having a singleton edge of equal weight consisting of that node.

We define the following auxiliary problem on hypergraphs (see also~\cite{mku_max}).

\defproblem{\HS{k}}{A weighted hypergraph $H=(V,E)$, with $E$ given as a list, and an integer~$k>0$.}{A subset $K$ of size $k$ of $V$ such that $H[K]$ has maximum weight.}

When $k=\cO(1)$, we preprocess the edges of $H$ as follows in order to have $\cO(1)$-time access to any queried edge.
We represent each edge as a string, whose letters correspond to its elements in increasing order.
Then, we sort all such strings lexicographically using radix sort in $\cO(|E|)$ time and construct a trie over them.
An edge can then be accessed in $\cO(k \log k)=\cO(1)$ time by a forward search starting from the root node of the trie.

A polynomial-time $\cO(n^{0.697831+\epsilon})$-approximation for \HS{k}, for any $\epsilon>0$, for the case when all hyperedges of $H$ have size at most 3 was shown in~\cite{mku_max} (see also~\cite{DBLP:journals/siamcomp/Applebaum13}).

Two remarks are in place. 
First, we can focus on edges of size up to $k$ as larger edges cannot, by definition, exist in any $k$-section. Second, \HS{k} is a generalization of the problem of deciding whether a $(c,k)$-hyperclique (i.e.~a set of $k$ nodes whose subsets of size $c$ are all in $E$) exists in a graph, which in turn is a generalization of \kCl. Unlike \kCl, the $(c,k)$-hyperclique problem is not known to benefit from fast matrix multiplication in general; see~\cite{LWW-soda18} for a discussion on its hardness. 

\begin{lemma}\label{lem:red_hyper}
\PPSMw can be reduced to \HS{k} for a hypergraph with $\ell$  nodes and $d$ edges in $\cO(N)$ time, where $N=d\ell$.  
\end{lemma}
\begin{proof}
We first compute the set $M_s$ of positions of mismatches of $q$ with each string $s \in \D$.
We ignore strings from $\D$ that match $q$ exactly, as they will match $q$ after changing any set of letters of $q$ to wildcards. This requires $\cO(d\ell)=\cO(N)$ time in total. 

Let us consider an empty hypergraph (i.e.~with no edges) $H$ on $\ell$ nodes, numbered $1$ through $\ell$. 
Then, for each string $s \in \D$, we add $M_s$ to the edge-set of $H$ if $|M_s| \leq k$; if this edge already exists, we simply increment its weight by $1$.

We set the parameter $k$ of \HS{k} to the parameter $k$ of \PPSMw. We now observe that for $K \subseteq V$ with $|K|=k$, the weight of $H[K]$ is equal to the number of strings that would match $q$ after replacing with wildcards the $k$ letters of $q$ at the positions corresponding to elements of $K$. The statement follows.
\end{proof}

An example of the reduction in \cref{lem:red_hyper} is shown in Figure~\ref{fig:redu-hkhyper}.

\begin{figure}[!t]
\begin{center}
\includegraphics[trim={0 0 0 0},clip,width=0.99\textwidth]{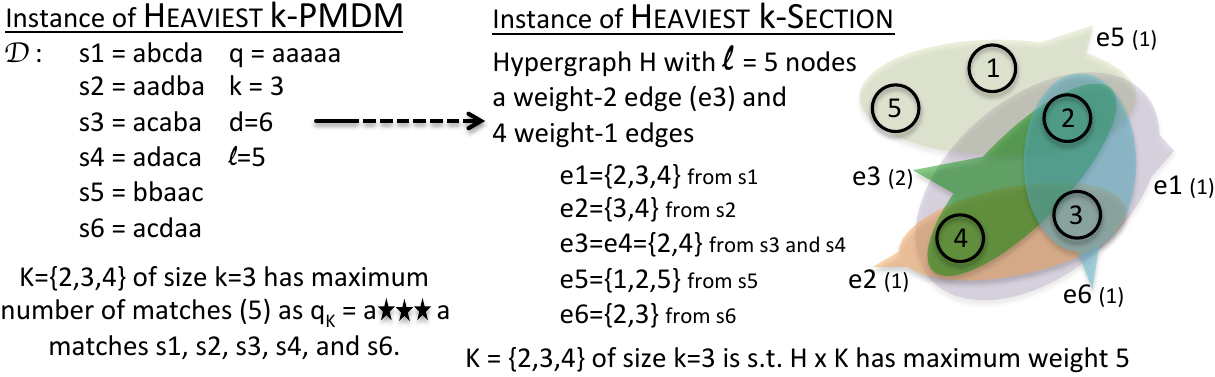}
\end{center}
\caption{An example of the reduction from \PPSMw to \HS{k}. The solutions are at the bottom. Each edge has its weight in brackets and the total weight is $d=6$.\label{fig:redu-hkhyper}}
\end{figure}

The next lemma gives a straightforward solution to \HS{k}. It is analogous to algorithm \textsc{Small}-$\ell$, presented in Section~\ref{sec:DS}, but without the optimization in computing sums of weights over subsets. It implies a linear-time algorithm for \HS{1}. 

\begin{lemma}\label{lem:Vk}
\HS{k}, for any constant $k$, can be solved in $\cO(|V|^k+|E|)$ time and $\cO(|V|+|E|)$ space.\label{lemma:heaviestksection}
\end{lemma}
\begin{proof}
For every subset $K\subseteq V$ of size at most $k$, we sum the weights of all edges corresponding to its subsets. There are $\genfrac(){0pt}{1}{|V|}{k}=\cO(|V|^k)$ choices for $|K|$, each having $2^k-1$  non-empty subsets: for every subset, we can access the corresponding edge (if it exists) in $\cO(1)$ time. 
\end{proof}

We next show that for the cases $k=2$ and $k=3$, there exist more efficient solutions. In particular,  we provide a linear-time algorithm for \HS{2}.

\begin{lemma}\label{lem:2}
\HS{2} can be solved in $\cO(|V|+|E|)$ time.
\end{lemma}
\begin{proof}
Let $K$ be a set of nodes of size $2$ such that $H[K]$ has maximum weight. We decompose the problem in two cases. 
For each of the cases, we give an algorithm that considers several $2$-sections such that the heaviest of them has weight equal to that of $H[K]$. 

\emph{Case 1.}~There is an edge $e=K$ in $E$. For each edge $e \in E$ of size $2$, i.e.~edge in the classic sense, we compute the sum of its weight and the weights of the nodes that it is incident to. 
This step requires $\cO(|E|)$ time.

\emph{Case 2.}~There is no edge equal to $K$ in $E$. We compute $H[\{v_1,v_2\}]$, where $v_1,v_2$ are the two nodes with maximum weight, i.e.~max and second-max. This step takes $\cO(|V|)$ time.

In the end, we return the heaviest $2$-section among those returned by the algorithms for the two cases, breaking ties arbitrarily.
\end{proof}

We next show that for $k=3$ the result of Lemma~\ref{lem:Vk} can be improved when $|E| = o(|V|^2)$.

\begin{lemma}\label{lem:3}
\HS{3} can be solved in time $\cO(|V|\cdot|E|)$ using $\cO(|V|+|E|)$ space.
\end{lemma}
\begin{proof}
Let $K$ be a set of nodes of size $3$ such that $H[K]$ has maximum weight. We decompose the problem into the following three cases.  

\emph{Case 1.}~There is an edge $e=K$ in $E$. We go through each edge $e \in E$ of size $3$ and compute the weight of $H[e]$ in $\cO(1)$ time. This takes $\cO(|E|)$ time in total. Let the edge yielding the maximum weight be $e_{\max}$.

\emph{Case 2.}~There is no edge of size larger than one in $H[K]$. We compute $H[\{v_1,v_2,v_3\}]$, where $v_1,v_2,v_3$ are the three nodes with maximum weight, i.e.~max, second-max and third-max. This step takes $\cO(|V|)$ time.

\emph{Case 3.}~There is an edge of size $2$ in $H[K]$. 
We can pick an edge $e$ of size $2$ from $E$ in $\cO(|E|)$ ways and a node $v$ from $V$ in $\cO(|V|)$ ways. We compute the weight of $H[(e\cup \{v\})]$ for all such pairs. Let the pair yielding maximum weight be $(e',u')$.

Finally, the maximum weight of $H[K']$ for $K' \in \{\,e_{\max},\,\{v_1,v_2,v_3\},\,e'\cup \{u'\}\,\}$ is equal to the weight of $H[K]$, breaking ties arbitrarily.
\end{proof}

We next address the remaining case of any arbitrarily large constant $k\ge 4$.

\begin{lemma}\label{lem:k}
\HS{k} for an arbitrarily large constant $k \ge 4$ can be solved in time $\cO((|V|\cdot|E|)^{k/3})$ using $\cO(|V|+|E|)$ space.
\end{lemma}
\begin{proof}
If $|E| > |V|^2$, then the simple algorithm of Lemma~\ref{lem:Vk} solves the problem in time \[\cO(|V|^k+|E|) = \cO(|V|^{k/3} (|V|^2)^{k/3}+|E|) = \cO((|V|\cdot|E|)^{k/3})\]
and linear space. We can thus henceforth assume that $|E| \leq |V|^2$.

Let $K$ be a set of nodes of size at most $k$ such that $H[K]$ has maximum weight.
If $H[K]$ contains isolated nodes (i.e.~nodes not contained in any edge),
they can be safely deleted without altering the result. We can thus assume that $H[K]$ does not contain isolated nodes, and that $|V| \le k|E|$ since otherwise the hypergraph $H$ would contain isolated nodes. 

We first consider the case that the rank $r(H[K])>1$, i.e.~there is an edge of $H[K]$ of size at least 2. We design a branching algorithm that constructs several candidate sets; the ones with maximum weight will have weight equal to that of $H[K]$. We will construct a set of nodes $X$, starting with $X:=\emptyset$.
For each set $X$ that we process, let $Z_X$ be the superset of $X$ of size at most $k$ such that $H[Z_X]$ has maximum weight.
We have the following two cases:

\emph{Case 1.}~There is an edge $e$ in $H[Z_X]$ that contains at least two nodes from $Z_X \setminus X$. To account for this case, we select every possible such edge $e$, set $X := X \cup e$, and continue the branching algorithm.

\emph{Case 2.}~Each edge in $H[Z_X]$ contains at most one node from $Z_X \setminus X$. In this case we conclude the branching algorithm as follows. For every node $v \in V \setminus X$ we compute its weight as the total weight of edges $Y \cup \{v\} \in E$ for $Y \subseteq X$ in $\cO(2^k) = \cO(1)$ time. Finally, in $\cO(|V|k) = \cO(|V|)$ time we select $k-|X|$ nodes with largest weights and insert them into $X$. The total time complexity of this step is $\cO(|V|)$. This case also works if $|X|=k$ and then its time complexity is only $\cO(1)$.

The correctness of this branching algorithm follows from an easy induction, showing that at every level of the branching tree there is a subset of $K$.

Let us now analyze the time complexity of this branching algorithm. Each branching in Case 1 takes $\cO(|E|)$ time and increases the size of $|X|$ by at least 2. At every node of the branching tree we call the procedure of Case 2. It takes $\cO(|V|)$ time if $|X| < k$.

If the procedure of Case 2 is called in a non-leaf node of the branching tree, then its $\cO(|V|)$ running time is dominated by the $\cO(|E|)$ time that is required for further branching since we have assumed that $|V| \le k|E|$. Hence, it suffices to bound (a) the total time complexity of calls to the algorithm for Case 2 in leaves that correspond to sets $X$ such that $|X| < k$ and (b) the total number of leaves that correspond to sets $X$ such that $|X| = k$.

If $k$ is even, (a) is bounded by $\cO(|E|^{(k-2)/2} |V|)$ and (b) is bounded by $\cO(|E|^{k/2})$. Hence, (b) dominates (a) and we have
\begin{equation}\label{eq:23}
\cO(|E|^{k/2}) = \cO(|E|^{k/3} |E|^{k/6}) = \cO(|E|^{k/3} |V|^{k/3}).
\end{equation}

If $k$ is odd, (a) is bounded by $\cO(|E|^{(k-1)/2} |V|)$ and (b) is bounded by $\cO(|E|^{(k-1)/2})$, which is dominated by (a). By using \eqref{eq:23} for $k-3$ we also have:

\begin{multline*}
\cO(|E|^{(k-1)/2} \cdot |V|) = \cO(|E|^{(k-3)/2} \cdot |E| \cdot |V|) =\\ \cO((|E|\cdot |V|)^{(k-3)/3} \cdot |E| \cdot |V|) = \cO((|E|\cdot |V|)^{k/3}).
\end{multline*}

We now consider the case that $r(H[K])=1$. We use the algorithm for Case 2 above that works in $\cO(|V|)$ time, which is $\cO(|V| \cdot |E|)$.
\end{proof}

Lemmas~\ref{lem:red_hyper}-\ref{lem:k} imply Theorem~\ref{the:kWMDM}, which we iteratively employ  to obtain the following result.

\begin{theorem}
\PPSM can be solved in time $\cO(N+\min\{N^{k/3},\ell^k\})$ using space $\cO(N)$ if $k=\cO(1)$, where $N=d\ell$.\label{the:PPSMZ}
\end{theorem}
\begin{proof}
We apply Lemma~\ref{lem:red_hyper} to obtain a hypergraph with $|V|=\ell$ and $|E|=d$.
Starting with $k=1$ and for growing values of $k$, we solve \HS{k} until we obtain a solution of weight at least $z$, employing either only Lemma~\ref{lem:Vk}, or Lemmas~\ref{lem:Vk},~\ref{lem:2},~\ref{lem:3},~\ref{lem:k} for $k=1,2,3$ and $k\geq 4$, respectively. We obtain $\cO(N+\min\{N^{k/3},\ell^k\})$ time and $\cO(N)$ space. 
\end{proof}

\section{Exact Algorithms for a Bounded Number $m$ of Query Strings}\label{sec:many_strings}

Recall that masking a potential match $(q_1,q_2)$ in record linkage can be performed by solving \PPSM twice and superimposing the wildcards (see Section~\ref{sec:intro}). In this section, we consider the following generalized version of \PPSM to perform the masking simultaneously. The advantage of this approach is that it minimizes the final number of wildcards in $q_1$ and $q_2$.

\defproblem{Multiple Pattern Masking for Dictionary Matching (\MPPSM)}{A dictionary $\D$ of $d$ strings, each of length $\ell$, a collection $\M$ of $m$ strings, each of length $\ell$, and a positive integer $z$.}{A smallest set $K\subseteq\{1,\ldots,\ell\}$ such that, for every $q$ from $\M$, $q_K\!=\!q\!\otimes\! K$ matches at least $z$ strings from $\D$.}

Let $N=d\ell$. We show the following theorem.

\begin{theorem}\label{thm:MPPSM}
\MPPSM can be solved in time $\cO(N+\min\{N^{k/3}z^{m-1},\ell^k\})$ if $k=\cO(1)$ and $m=\cO(1)$, where $N=d\ell$.\label{the:MPPSMZ}
\end{theorem}

We use a generalization of \HS{k} in which the weights are $m$-tuples that are added and compared component-wise, and we aim to find a subset $K$ such that the weight of $H[K]$ is at least $(z,\ldots,z)$. An analogue of Lemma~\ref{lem:Vk} holds without any alterations, which accounts for the $\cO(N+\ell^k)$-time algorithm. We adapt the proof of~\cref{lem:k} as follows. The branching remains the same, but we have to tweak the final step, that is, what happens when we are in Case 2. For $m=1$ we could simply select a number of largest weights, but for $m>1$ multiple criteria need to be taken into consideration. All in all, the problem reduces to a variation of the classic Multiple-Choice Knapsack problem~\cite{Kellerer2004}, which we solve using dynamic programming. 

The variation of the classic Multiple-Choice Knapsack problem is as follows.

\defproblem{$\kappa$ Heaviest Vectors ($\kappa$-\HV)}{A collection $\mathcal{T}$ of $t$ vectors from $\mathbb{Z}_{\geq 0}^m$, a vector $x$ from $\{0,\ldots,z\}^m$, for a positive integer $z$, and an integer $\kappa \in \{0,\ldots,t\}$.}{Compute $\kappa$ elements of $\mathcal{T}$ (if they exist) such that if $y$ is their component-wise sum, $y[i]\geq x[i]$ for all $i \in \{1,\ldots,m\}$.}

The exact reduction from Case 2 is as follows: the set $\mathcal{T}$ contains weights of subsequent nodes $v \in V \setminus X$ (defined as the sums of weights of edges $Y \cup \{v\} \in E$ for $Y \subseteq X$), so $t \le |V|$, $x$ is $(z,\ldots,z)$ minus the sum of weights of all edges $e \in E$ such that $e \subseteq X$, and $\kappa=k-|X|$.

The solution to $\kappa$-\HV is a rather straightforward dynamic programming.

\begin{lemma}
For $\kappa,m=\cO(1)$, $\kappa$-\HV can be solved in time $\cO(t \cdot z^{m-1})$.
\end{lemma}
\begin{proof}
We apply dynamic programming. Let $\mathcal{T}=v_1,\ldots,v_t$. We compute an array $A$ of size $\cO(t\kappa z^{m-1})$ such that, for $i \in \{0,\ldots,t\}$, $j \in \{0,\ldots,\kappa\}$ and $v \in \{0,\ldots,z\}^{m-1}$,
\[A[i,j,v]=\max\{a\,:\,\exists S \subseteq \{v_1,\ldots,v_i\},\,|S|=j,\,\sum_{u \in S} u=(v,a)\},\]
where $(v,a)$ denotes the operation of appending element $a$ to vector $v$.
From each state $A[i,j,v]$ we have two transitions, depending on whether $v_{i+1}$ is taken to the subset or not. Each transition is computed in $\cO(m)=\cO(1)$ time. This gives time $\cO(t \cdot z^{m-1})$ in total.

The array is equipped with a standard technique to recover the set $S$ (parents of states). The final answer is computed by checking, for each vector $v \in \{0,\ldots,z\}^{m-1}$ such that $v[i] \ge x[i]$, for all $i=1,\ldots,m-1$, if $A[t,\kappa,v] \ge x[m]$.
\end{proof}

Overall, we pay an additional $\cO(z^{m-1})$ factor in the complexity of handling of Case~2, which yields the complexity of Theorem~\ref{thm:MPPSM}.

\section{A Data Structure for \PPSM Queries}\label{sec:DS} 

We next show algorithms and data structures for the $\PPSM$ problem under the assumption that $2^\ell$ is reasonably small.  
We measure space in terms of $w$-bit machine words, where $w=\Omega(\log (d\ell))$, and focus on showing space vs.~query-time trade-offs for answering $q,z$ $\PPSM$ queries over $\D$. A summary of the complexities of the data structures is shown in~\cref{tab:DSs}. Specifically, algorithm {\sc Small-$\ell$} and data structure {\sc Simple} are used as building blocks in the more involved data structure {\sc Split} underlying the following theorem.

\begin{theorem}\label{the:DS}
There exists a data structure that answers $q,z$ $\PPSM$ queries over $\D$ in time $\cO(2^{\ell/2}(2^{\ell/2}+\tau)\ell)$ w.h.p.~and requires space $\cO(2^{\ell}d^2/\tau^2+2^{\ell/2}d)$, for any $\tau\in[1,d]$.
\end{theorem}

\paragraph{Algorithm {\sc Small-$\ell$}: $\cO(d\ell)$ Space, $\cO(2^\ell\ell+d\ell)$ Query Time.}

Our algorithm is based on the Fast zeta/Möbius transform~\cite[Theorem 10.12]{DBLP:books/sp/CyganFKLMPPS15}. No data structure on top of the dictionary $\D$ is stored. In the query algorithm, we initialize an integer array $A$ of size $2^\ell$ with zeros. 
For an $\ell$-bit vector $m$, by $K_m \subseteq\{1,\ldots,\ell\}$ let us denote the set of the positions of set bits of $m$.
Now for every possible $\ell$-bit vector $m$ we want to compute the number of strings in $\D$ that match $q_{K_m}=q\otimes K_m$. 

To this end, for every string $s \in\D$, we compute the set $K$ of positions in which $s$ and $q$ differ. For $m$ that satisfies $K=K_m$, we increment $A[m]$, where $m$ is the integer representation of the bit vector.
This computation takes $\cO(d\ell)$ time and $\cO(1)$ extra space. 
Then we apply a folklore dynamic-programming-based approach 
to compute an integer array $B$, which is defined as follows: \[B[m]=\sum_{j \in S(m)} A[j] \text{, where }  S(m)=\{j \in [1,2^\ell] : K_j \subseteq K_m\}.\] 
In other words, $B[m]$ stores the number of strings from $\D$ that match $q_{K_m}$. 

We provide a description of the folklore algorithm here for completeness.
Consider a vector (mask) $m$. Let $S(m,i)$ consist of the subsets of $m$ which do not differ from $m$ but (possibly) in the rightmost $i$ bits, and
\[B[m,i]=\sum_{j \in S(m,i)} A[j].\] 
Clearly, $S(m)$ is equal to $S(m,\ell)$, and hence $B[m]$ is equal to $B[m,\ell]$. 
The following equation is readily verified (in the first case, since the $i$th bit of $m$ is $0$, no element of $S(m)$ can have the $i$th bit set):
\[B[m,i]= \begin{cases}
	B[m,i-1] &  \text{if the $i$th bit of $m$ is $0$,} \\
    B[m,i-1]\textsf{ OR } B[m \textsf{ XOR } 2^i,i-1] & \text{if the $i$th bit of $m$ is $1$.}
\end{cases}\]

By \textsf{OR} and \textsf{XOR} we denote the standard bitwise operations. Overall, there are $\cO(2^\ell \ell)$ choices for $m$ and $i$. We can compute $B[\cdot,\cdot]$ column by column, in constant time per entry, thus obtaining an $\cO(2^\ell \ell)$-time algorithm.
We can limit the space usage to $\cO(2^\ell)$ by discarding column $i$ when we are done computing column $i+1$.

Thus, overall, the (query) time required by algorithm \textsc{Small-$\ell$} is $\cO(\ell2^\ell\!+\!d\ell)$, the data structure space is $\cO(d\ell)$, and the extra space is $\cO(2^\ell)$.\\

We now present \textsc{Simple}, an auxiliary data structure, which we will apply later on to construct DS \textsc{Split}, a data structure with the space/query-time trade-off of Theorem~\ref{the:DS}.

\paragraph{DS {\sc Simple}: $\cO(2^\ell d)$ Space, $\cO(2^\ell \ell)$ Query Time.}
We initialize an empty set $\mathcal Q$. For each possible subset of $\{1,\ldots,\ell\}$ we do the following.
We mask the corresponding positions in all strings from $\mathcal{D}$ and then sort the masked strings lexicographically.
By iterating over the lexicographically sorted list of the masked strings, we count how many copies of each distinct (masked) string we have in our list.
We insert each such (masked) string to $\mathcal Q$ along with its count.
After processing all $2^\ell$ subsets, we construct a compacted trie for the strings in $\mathcal Q$; each leaf corresponds to a unique element of $\mathcal Q$, and stores this element's count.
The total space occupied by this compacted trie is thus $\cO(2^\ell d)$.
Upon an on-line query $q$ (of length $\ell$) and $z$, we apply all possible $2^\ell$ masks to $q$ and read the count for each of them from the compacted trie in $\cO(\ell)$ time per mask. 
Next, we show how to decrease the exponential dependency on $\ell$ in the space complexity when $2^\ell =o(d)$, incurring extra time in the query.

\begin{table}[t]
    \centering
    \begin{tabular}{c|c|c}
    Data structure & Space & Query time \\\hline
    Algorithm \textsc{Small-$\ell$} & $\cO(d\ell)$ & $\cO(2^\ell\ell+d\ell)$ \\
    DS {\sc Simple} & $\cO(2^\ell d)$ & $\cO( 2^\ell \ell)$  \\
    DS {\sc Split}, any $\tau$ &  $\cO(2^{\ell}d^2/\tau^2+2^{\ell/2}d)$ & $\cO(2^{\ell/2} \cdot (2^{\ell/2}+\tau)\ell)$ \\
    DS {\sc Split} for $\tau=2^{\ell/4}\sqrt{d}$ &  $\cO(2^{\ell/2}d)$ & $\cO(2^\ell\ell+2^{3\ell/4}\sqrt{d}\ell)$
    \end{tabular}
    \caption{Basic complexities of the data structures from~\cref{sec:DS}.}
    \label{tab:DSs}
\end{table}

\paragraph{DS {\sc Split}: $\cO(2^{\ell}d^2/\tau^2+2^{\ell/2}d)$ Space, $\cO(2^{\ell/2} \cdot (2^{\ell/2}+\tau)\ell)$ Query Time, for any $\tau$.}
This trade-off is relevant when $\tau =\omega(\sqrt{d})$; otherwise the DS \textsc{Simple} is better. 

We split each string $p \in \D$ roughly in the middle, to prefix $p_{L}$ and suffix $p_{R}$; specifically, $p=p_Lp_R$ and $|p_L|=\lceil \ell/2 \rceil$.
We create dictionaries $\D_L\!=\!\{p_L \!:\! p\! \in\! \D\}$ and $\D_R\!=\!\{p_R \!:\! p\! \in\! \D\}$.
Let us now explain how to process $\D_L$; we process $\D_R$ analogously.
Let $\lambda\!=\!\lceil \ell/2 \rceil$.
We construct DS \textsc{Simple} over $\D_L$. 
This requires space $\cO(2^{\ell/2}d)$. Let $\tau$ be an input parameter, intuitively used as the minimum frequency threshold. For each of the possible $2^\lambda$ masks, we can have at most $\lfloor d/\tau \rfloor$ (masked) strings with frequency at least $\tau$. 
Over all masks, we thus have at most $2^\lambda \lfloor d/\tau \rfloor$ such strings, which we call \emph{$\tau$-frequent}. 
For every pair of $\tau$-frequent strings, one from $\D_L$ and one from $\D_R$, we store the number of occurrences of their concatenation in $\D$ using a compacted trie as in DS \textsc{Simple}. 
This requires space $\cO(2^{\ell}d^2/\tau^2)$.

Consider $\D_{L}$. For each mask $i$ and each string $p_L\in \D_{L}$, we can afford to store the list of all strings in $\D_L$ that match $p_L \otimes i$. Note that we have computed this information when sorting for constructing DS \textsc{Simple} over $\D_L$. This information requires space $\cO(2^{\ell/2}d)$. Thus, DS \textsc{Split} requires $\cO(2^{\ell}d^2/\tau^2+2^{\ell/2}d)$ space overall.

Let us now show how to answer an on-line $q,z$ query. 
Let $q=q_Lq_R$ with $|q_L|=\lceil \ell/2 \rceil$.
We iterate over all possible $2^{\ell}$ masks. 

For a mask $i$, let $q'=q\otimes i$.
We split $q'$ into two halves, $q'_L$ and $q'_R$  with $q'=q'_Lq'_R$ and $|q'_L|=\lceil \ell/2 \rceil$.
First, we check whether each of $q'_L$ and $q'_R$ is $\tau$-infrequent using the DS \textsc{Simple} we have constructed for $\D_L$ and $D_R$, respectively, in time $\cO(\ell)$. We have the following two cases (inspect also Figure~\ref{fig:split}).
\begin{itemize}
    \item If both halves are $\tau$-frequent, then we can read the frequency of their concatenation using the stored compacted trie in time $\cO(\ell)$.
    \item Else, at least one of the two halves is $\tau$-infrequent. Assume without loss of generality that $q'_L$ is $\tau$-infrequent. Let $\mathcal F$ be the dictionary consisting of at most $\tau$ strings from $\D_R$ that correspond to the right halves of strings in $\D_L$ that match $q'_L$. 
    Na\"{\i}vely counting how many elements of $\mathcal F$ match $q'_R$ could require $\Omega(\tau \ell)$ time, and thus $\Omega(2^\ell \tau \ell)$ overall.
    Instead, we apply algorithm \textsc{Small}-$\ell$ on $q_R$ and $\mathcal F$. 
    The crucial point is that if we ever come across $q'_L$ again (for a different mask on $q$), we will not need to do anything. 
    We can maintain whether $q'_L$ has been processed by temporarily marking the leaf corresponding to it in DS {\sc Simple} for $\D_L$.
    Thus, overall, we perform the \textsc{Small}-$\ell$ algorithm $\cO(2^{\ell/2})$ times, each time in $\cO((2^{\ell/2}+\tau)\ell)$ time.
    This completes the proof of Theorem~\ref{the:DS}. 
\end{itemize}

\paragraph{Efficient Construction.} For completeness, we next show how to construct DS \textsc{Split} in $\cO(d\ell \log (d\ell)+2^\ell d \ell+2^{\ell}\ell d^2/\tau^2)$ time. We preprocess $\mathcal{D}$ by sorting its letters in $\cO(d\ell \log (d\ell))$ time. The DS \textsc{Simple} for $\D_L$ and $\D_R$ can then be constructed in $\cO(2^{\ell/2}d \ell)$ time. We then create the compacted trie for pairs of $\tau$-frequent strings. 
For each of the $2^\ell$ possible masks, say $i$, and each string $p \in \D$, we split $p'=p\otimes i$ in the middle to obtain $p'_L$ and $p'_R$. If both $p'_L$ and $p'_R$ are $\tau$-frequent then $p'$ will be in the set of strings for which we will construct the compacted trie for pairs of $\tau$-frequent strings. 
The counts for each of those strings can be read in $\cO(\ell)$ time from a DS \textsc{Simple} over $\D$, which we can construct in time $\cO(2^{\ell}d \ell)$---this data structure is then discarded. The compacted trie construction requires time $\cO(2^{\ell}\ell d^2/\tau^2)$.

\begin{figure}[!t]
        \begin{center}
        \includegraphics[trim={0 0 0 0},clip,width=0.85\textwidth]{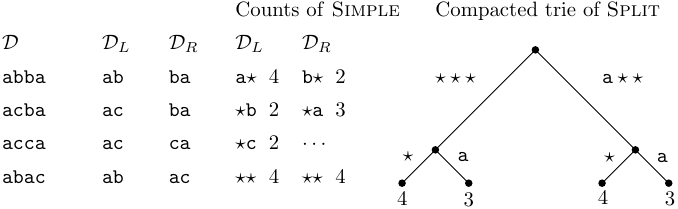}
        \end{center}
        \caption{Let $\tau=3$. If both $q'_L$ and $q'_R$ are $3$-frequent (we check this using the counts of DS {\sc Simple}), then we read the count for $q'_Lq'_R$ from the compacted trie of DS {\sc Split}. If $q'_L$ is $3$-infrequent, then we apply \textsc{Small-$\ell$} on $q_R$ and on the dictionary consisting of at most $\tau=3$ strings from $\mathcal{D}_R$ corresponding to the right halves of strings in $\mathcal{D}_L$ that match $q'_L$.}
        \label{fig:split}
\end{figure}

\paragraph{Comparison of the Data Structures.} DS {\sc Simple} has lower query time than algorithm {\sc Small-$\ell$}. However, its space complexity can be much higher. DS {\sc Split} can be viewed as an intermediate option. For $\tau$ as in~\cref{tab:DSs}, it has lower query time than algorithm {\sc Small-$\ell$} for $d=\omega(2^{3\ell/2})$, while keeping moderate space complexity. DS {\sc Split} always has higher query time than DS {\sc Simple}, but its space complexity is lower by a factor of $2^{\ell/2}$. For example, for $d=2^{2\ell}$ we get the complexities  shown in~\cref{tab:DSs2l}.

\begin{table}[htpb]
    \centering
    \begin{tabular}{c|c|c}
    Data structure & Space & Query time \\\hline
    Algorithm \textsc{Small-$\ell$} & $\cO(2^{2\ell}\ell)$ & $\cO(2^{2\ell}\ell)$ \\
    DS {\sc Simple} & $\cO(2^{3\ell})$ & $\cO(2^\ell \ell)$  \\
    DS {\sc Split} for $\tau=2^{5\ell/4}$ &  $\cO(2^{5\ell/2})$ & $\cO(2^{7\ell/4}\ell)$
    \end{tabular}
    \caption{Basic complexities of the data structures from~\cref{sec:DS} for $d=2^{2\ell}$.}
    \label{tab:DSs2l}
\end{table}

Let us now discuss why our data structure results cannot be directly obtained using the same data structures as for the problem Dictionary Matching with $k$-wildcards (see~\cref{sec:intro} for the problem definition).  Conceivably, one could construct such a data structure, and then iterate over all subsets of $\{1, \ldots, \ell\}$, querying for the masked string. Existing data structures for dictionary matching with wildcards (cf.~\cite[Table 1]{DBLP:journals/mst/BilleGVV14}, \cite{LNV14-stacs}, and~\cite{DBLP:conf/esa/GawrychowskiLN14}), that allow querying a pattern with at most $\ell$ wildcards, have:
\begin{enumerate}
    \item[(a)] either $\Omega(\min\{\sigma^\ell, d\})$ query time, thus yielding $\Omega(2^\ell \cdot \min\{\sigma^\ell, d\})$ query time for our problem, and space $\Omega(d\ell)$, a trade-off dominated by the \textsc{Small}-$\ell$ algorithm (cf.~our~\cref{tab:DSs});
    \item[(b)] or $\Omega(\ell)$ query time, thus yielding $\Omega(2^\ell \ell)$ query time for our problem, and $\Omega(d \ell \log^\ell \log (d\ell))$ space, a trade-off dominated by the DS \textsc{Simple} (cf.~our~\cref{tab:DSs}).
\end{enumerate}

\section{Approximation Algorithm for \PPSM}\label{sec:approx}

Clearly, \PPSM is at least as hard as \PPSMz
because it also outputs the positions of the wildcards (set $K$). Thus, \PPSM is also NP-hard. In what follows, we show existence of a polynomial-time approximation algorithm for \PPSM whose  approximation factor is given with respect to $d$. Specifically, we show the following approximation result for \PPSM.

\begin{theorem}\label{the:approx}
For any constant $\epsilon>0$, there is an $\cO(d^{1/4+\epsilon})$-approximation algorithm for \PPSM, whose running time is polynomial in $N$, where $N=d\ell$.
\end{theorem}

Our result is based on a reduction to the Minimum Union (\MU) problem~\cite{mku_min}, which we define next.

\defproblem{Minimum Union (\MU)}{A collection $\mathcal{S}$ of $d$ sets over a universe $U$ and a positive integer $z\leq d$.}{A collection $\mathcal{T} \subseteq \mathcal{S}$ with $|\mathcal{T}|=z$ such that the size of $\cup_{S\in \mathcal{T}}S$ is minimized.}

To illustrate the \MU problem, consider an instance of it with \[U=\{1,2,3,4,5\} \text{ and } \mathcal{S}=\{\{1\},\{1,2,3\},\{1,3,5\},\{3\},\{3,4,5\},\{4\},\{4,5\},\{5\}\},\] where $d=|\mathcal{S}|=8$, and $z=4$. Then $\mathcal{T}=\{\{3\},\{3,4,5\},\{4\},\{4,5\}\}$ is a solution because $|\mathcal{T}|=z=4$ and $|\cup_{S\in \mathcal{T}}S|=3$ is minimum. The \MU problem is NP-hard and the following approximation result is known.

\begin{theorem}[\cite{mku_min}] 
For any constant $\epsilon>0$, there is an $\cO(d^{1/4+\epsilon})$-approximation algorithm for \MU, whose running time is polynomial in the size of $\mathcal{S}$.\label{thm:mkuapprox}
\end{theorem}

We next describe the reduction that leads to our result.

\begin{theorem}
\PPSM can be reduced to \MU in time polynomial in $N$.\label{thm:PPSMtoMkU}
\end{theorem}

\begin{proof}
We reduce the \PPSM problem to \MU in polynomial time as follows. Given any instance $\mathcal{I}_{\PPSM}$ of \PPSM, we construct an instance  $\mathcal{I}_{\text{MU}}$ of \MU  in time $\cO(d\ell)$ by performing the following steps: 
\begin{enumerate}
\item The universe $U$ is set to $\{1,\ldots, \ell\}$. 
\item We start with an empty collection $\mathcal{S}$. Then, for each string $s_i$ in $\mathcal{D}$, we add member $S_i$ to $\mathcal{S}$, where $S_i$ is the set of positions where string $q$ and string $s_i$ have a mismatch. This can be done trivially in time $\cO(d\ell)$ for all strings in $\mathcal{D}$.  
\item Set the $z$ of the \MU problem to the $z$ of the \PPSM problem. 
\end{enumerate}

Thus, the total time $\cO(d\ell)$ needed for Steps 1 to 3 above is clearly polynomial in the size of $\mathcal{I}_{\PPSM}$. 

\begin{claim}\label{claim:eq}
For any solution $\mathcal{T}$ to $\mathcal{I}_{\text{MU}}$ and any solution $K$ to $\mathcal{I}_{\PPSM}$, such that $\mathcal{I}_{\text{MU}}$ is obtained from $\mathcal{I}_{\text{\PPSM}}$ using the above three steps, we have $|K|=|\cup_{S\in \mathcal{T}}S|$.
\end{claim}
\begin{proof}[of Claim]
Let $\mathcal F \subseteq \D$ consist of $z$ strings that match $q_K$. Further, let the set $\mathcal F^*$ consist of the elements of $\mathcal S$ corresponding to strings in $\mathcal F$. We have
$|\cup_{S\in \mathcal{T}}S|\leq |\cup_{S\in \mathcal{F^*}}S|\leq |K|$.

Now, let $C=\cup_{S\in \mathcal{T}}S$. Then, $q_C\!=\!q\!\otimes\! C$ matches at least $z$ strings from $\D$ and hence $|K| \leq |C|= |\cup_{S\in \mathcal{T}}S|$.
\end{proof}

To conclude the proof, it remains to show that given a solution $\mathcal{T}$ to $\mathcal{I}_{\text{MU}}$ we can obtain a solution $K$ to $\mathcal{I}_{\PPSM}$ in time polynomial in the size of $\mathcal{I}_{\text{MU}}$.
This readily follows from the proof of the above claim: it suffices to set $K=\cup_{S\in \mathcal{T}}S$. 
\end{proof}

We can now prove Theorem~\ref{the:approx}.

\begin{proof}[of Theorem~\ref{the:approx}.] 
The reduction in Theorem~\ref{thm:PPSMtoMkU} implies that there is a polynomial-time approximation algorithm for \PPSM. In particular, Theorem~\ref{thm:mkuapprox} provides an approximation guarantee for \MU that depends on the number of sets of the input $\mathcal{S}$. In Step 2 of the reduction of Theorem~\ref{thm:PPSMtoMkU}, we construct \emph{one} set for the \MU instance per \emph{one} string of the dictionary $\mathcal{D}$ of the \PPSM instance. Also, from the constructed solution $\mathcal{T}$ to the \MU instance, we obtain a solution $K$ to the \PPSM instance by simply substituting the positions of $q$ corresponding to the elements of the sets of $\mathcal{T}$ with wildcards. This construction implies the approximation result of Theorem~\ref{the:approx} that depends on the size of $\mathcal{D}$.
\end{proof}

Applying Theorem~\ref{the:approx} to solve \PPSM is not practical, as in real-world applications, such as those in~\cref{sec:intro}, $d$ is typically in the order of thousands or millions~\cite{Christen2020lsd,eatkde,ntoulas,sigir17}.

\paragraph{Sanity Check.} We remark that Theorem~\ref{thm:red} (reduction from \kCl to \PPSMzw) and Theorem~\ref{the:approx} (approximation algorithm for \PPSM) do not contradict the inapproximability results for the maximum clique problem (see Section~\ref{sec:hardness}), since our reduction from \kCl to \PPSMzw cannot be adapted to a reduction from maximum clique to \PPSMz.

\paragraph{Two-Way Reduction.} Chlamt\'{a}\v{c} et al.~\cite{mku_min} additionally show that their polynomial-time $\cO(d^{1/4+\epsilon})$-approximation algorithm for \MU is tight under a plausible conjecture for the so-called Hypergraph Dense vs Random problem. In what follows, we also show that approximating the \MU problem can be reduced to approximating \PPSM in polynomial time and hence the same tightness result applies to \PPSM.

\begin{theorem}
\MU can be reduced to \PPSM in time polynomial in the size of $\mathcal{S}$.
\end{theorem}
\begin{proof} 
Let $||\mathcal{S}||$ denote the total number of elements in the $d$ members of $\mathcal{S}$.
We reduce the \MU problem to the \PPSM  problem in polynomial time as follows. 
Given any instance $\mathcal{I}_{\text{\MU}}$ of \MU, we construct an instance $\mathcal{I}_{\PPSM}$ of \PPSM by performing the following steps: 

\begin{enumerate}
\item  Sort the union of all elements of members of $\mathcal{S}$, assign to each element $j$ a unique rank $\textsf{rank}(j) \in \{1,\ldots,|U|\}$, and set $\ell=|U|$. This can be done in $\cO(||\mathcal{S}|| \log ||\mathcal{S}||)$ time.
\item Set the query string $q$ equal to the string $\texttt{a}^\ell$ of length $\ell$.
For each set $S_i$ in $\mathcal{S}$, construct a string $s_i=\texttt{a}^\ell$, set $s_i[\textsf{rank}(j)]:=\texttt{b}$ if and only if $j\in S_i$, and add $s_i$ to dictionary $\D$. This can be done in $\cO(d\ell)$ time.
\item Set the $z$ of the \PPSM  problem equal to the $z$ of the \MU problem. This can be done in $\cO(1)$ time.
\end{enumerate}

Thus, the total time $\cO(d\ell \log(d \ell))$ needed for Steps 1 to 3 above is clearly polynomial in the size of $\mathcal{I}_{\MU}$ as $\ell\leq ||\mathcal{S}||$. 

A proof of the following claim is analogous to that of~\cref{claim:eq}.
\begin{claim}
For any solution $\mathcal{T}$ to $\mathcal{I}_{\text{MU}}$ and any solution $K$ to $\mathcal{I}_{\PPSM}$, such that $\mathcal{I}_{\text{\PPSM}}$ is obtained from $\mathcal{I}_{\text{MU}}$ using the above three steps, we have $|K|=|\cup_{S\in \mathcal{T}}S|$.
\end{claim}

To conclude the proof, it remains to show that, given a solution
$K$ to $\mathcal{I}_{\PPSM}$, we can obtain a solution $\mathcal{T}$ to $\mathcal{I}_{\text{MU}}$ 
in time polynomial in the size of $\mathcal{I}_{\PPSM}$.
It suffices to pick $z$ sets in $\mathcal{S}$ that are subsets of $K$.
Their existence is guaranteed by construction,
because such sets correspond to the at least $z$ strings in $\D$ that have \texttt{b} in a subset of the positions in $K$.
This selection can be done na\"{\i}vely in $\cO(||\mathcal{S}||)$ time.
Finally, the above claim guarantees that they indeed form a solution to $\mathcal{I}_{\text{MU}}$. 

\end{proof}

\section{A Greedy Heuristic for \PPSM}\label{sec:greedy}

We design a heuristic called \Greedy that solves \PPSM and which, for a given constant $\tau\geq 1$, iteratively applies Theorem~\ref{the:kWMDM} (see~\cref{sec:exact}), for $k=1,\ldots,\tau$. Intuitively, the larger the $\tau$, the more effective -- but the slower -- \Greedy is. Specifically, in iteration $i=1$, we apply Theorem~\ref{the:kWMDM} for $k=1,\ldots,\tau$ and check whether there are at least $z$ strings from $\mathcal{D}$ that can be matched when at most $k$ wildcards are substituted in the query string $q$. If there are, we return the minimum such $k$ and terminate. Clearly, by Theorem~\ref{the:PPSMZ}, the returned solution $K_1$ is an optimal solution to \PPSM. Otherwise, we proceed into the next iteration, $i=2$. In this iteration, we construct string $q_{K_1}=q \otimes K_1$ and apply Theorem~\ref{the:kWMDM}, for $k=1,\ldots,\tau$, to $q_{K_1}$. This returns a solution $K_2$ telling us whether there are at least $z$ strings from $\mathcal{D}$ that can be matched with $q_{K_2}=q_{K_1}\otimes K_2$. If there are, we return $K_1\cup K_2$, which is now a (sub-optimal) solution to \PPSM, and terminate. Otherwise, we proceed into iteration $i=3$, which is analogous to iteration $i=2$. Note that \Greedy always terminates with some (sub-optimal) solution $K_1\cup K_2\cup \cdots \cup K_j$, for some $j\leq \lceil \ell/\tau \rceil$. Namely, in the worst case, it returns set $\{1,\ldots,\ell\}$ after $\lceil \ell/\tau \rceil$ iterations and $q_{\{1,\ldots,\ell\}}$ matches all strings in $\mathcal{D}$. The reason why \Greedy does not guarantee finding an optimal solution to \PPSM is that at iteration $i$ {\em we fix} the positions of wildcards based on solution $K_1\cup\cdots \cup K_{i-1}$, whereas some of those positions might not belong to the global optimal solution. 

Since $\tau=\cO(1)$, the time complexity of \Greedy is $\cO((N+N^{\tau/3})\ell)$: each iteration takes time $\cO(N+N^{\tau/3})$ by Theorem~\ref{the:kWMDM}, and then there are no more than $\lceil \ell/\tau \rceil=\cO(\ell)$ iterations. The space complexity of \Greedy is $\cO(N)$. The hypergraph $H=(V,E)$ used in the implementation of Theorem~\ref{the:kWMDM} has edges of size up to $k$. If every string in $\mathcal{D}$ has more than $k$ mismatches with $q$, then all edges in $H$ have size larger than $k$. In this case, we preprocess the hypergraph $H$, as detailed below. The objective is to remove selected nodes and edges from $H$, so that it has at least one edge of size up to $k$ and then apply \Greedy. 

\paragraph{Hypergraph Preprocessing.} Let us now complete the description of our heuristic, by describing the hypergraph preprocessing. We want to ensure that hypergraph $H=(V,E)$ has at least one edge of size up to $k$ so that \Greedy can be applied. To this end, if there is no edge of size up to $k$ at some iteration, then we add some nodes into the partial solution with the following heuristic. 
\begin{enumerate} 
\item We assign a score $s(u)$ to each node $u$ in $V$ using the function: 

$$s(u)=|E_u|\cdot \frac{\sum_{e\in E_u}w(e)}{\sum_{e\in E_u}|e|},$$

\noindent where $E_u=\{e \in E: u \in e\}$ and $w(e)$ is the edge weight. 
\item Then, we add the node with maximum score from $H$ (breaking ties arbitrarily) into the partial solution and update the edges accordingly. 
\end{enumerate}
These two steps are repeated until there is at least one edge of size up to $k$; this takes $\cO(d\ell^2)$ time. 
After that, we add the removed nodes into the current solution $K_{k}$ and use the resulting hypergraph to apply \Greedy. \\

The intuition behind the above process is to add  nodes which appear in many short edges (so that we mask few positions) with large weight (so that the masked positions greatly increase the number of matched strings). We have also tried a different scoring function  $s'(u)=\sum_{e\in E_u}\frac{w(e)}{|e|}$ instead of $s(u)$, but the results were worse, and thus not reported.

\section{Experimental Evaluation}\label{sec:exp}

\paragraph{Methods.} We compared the performance of our heuristic \Greedy (henceforth, \textsc{GR} $\tau$), for the values $\tau\in[3,5]$, 
for which its time complexity is subquadratic in $N$,
to the following two algorithms:

\begin{itemize}
\item{\textsc{Baseline} (henceforth, \textsc{BA}).}~In iteration $i$, \textsc{BA} adds a node of hypergraph $H$ into $K$ and updates $H$ according to the preprocessing described in Section~\ref{sec:greedy}. If at least $z$ strings from $\mathcal{D}$ match the query string $q$ after the positions in $q$ corresponding to $K$ are replaced with wildcards, \textsc{BA} returns $K$ and terminates; otherwise, it proceeds into iteration $i+1$. Note that 
\textsc{BA} generally constructs a {\em suboptimal} solution $K$ to {\PPSM} and takes $\cO(d\ell^2)$ time.

\item{\textsc{Bruteforce} (henceforth, \textsc{BF}).}~In iteration $i$, \textsc{BF}  applies Lemma~\ref{lemma:heaviestksection} in the process of obtaining an \emph{optimal} solution $K$ of size $i=k$ to \PPSM. In particular, it checks whether at least $z$ strings from $\mathcal{D}$ match the query string $q$, after the $i$ positions in $q$ corresponding to $K$ are replaced with wildcards. If the check succeeds,  \textsc{BF} returns $K$ and terminates; otherwise, it proceeds into iteration $i+1$. \textsc{BF} takes $\cO(k(2\ell)^k+dk)$ time (see Lemma~\ref{lemma:heaviestksection}). 
\end{itemize}

Since -- as mentioned in Section~\ref{sec:intro} -- there are no existing algorithms for addressing \PPSM, in the evaluation we used our own baseline \textsc{BA}. We have implemented all of the above algorithms in \texttt{C++}. Our implementations are freely available at \url{https://bitbucket.org/pattern-masking/pmdm/}.

\paragraph{Datasets.} We used the North Carolina Voter Registration database~\cite{ncvr} (\texttt{NCVR}); a standard benchmark dataset  for record  linkage~\cite{peteracmiq16,eatkde,Karapiperis2019dpd,Vatsalan2014cikm}. \texttt{NCVR} is a collection of 7,736,911  records with attributes such as \emph{Forename}, \emph{Surname}, \emph{City}, \emph{County},  and \emph{Gender}. We generated 4 subcollections of \texttt{NCVR}: (I) \texttt{FS} is comprised of all 952,864 records having  \emph{Forename} and \emph{Surname} of total length $\ell=15$; (II) \texttt{FCi} is comprised of all 342,472 records having \emph{Forename} and \emph{City} of total length $\ell=15$; (III) \texttt{FCiCo} is comprised of all  342,472 records having \emph{Forename},  \emph{City}, and \emph{County}  of  total length $\ell=30$; and (IV) {\texttt{FSCiCo}} is comprised of all 8,238  records  having \emph{Forename}, \emph{Surname}, \emph{City} and \emph{County} of total length $\ell=45$. 

We also generated a synthetic dataset, referred to as \texttt{SYN}, using the IBM Synthetic Data  Generator~\cite{generator}, a standard tool for generating sequential datasets~\cite{ibmsd1,ibmsd2}. \texttt{SYN} contains a collection of $6\cdot 10^6$ records, each of length $\ell=50$, over an alphabet of size $|\Sigma|=10$. We also generated
subcollections of \texttt{SYN} comprised of: $x\cdot 10^6$ arbitrarily selected records; the length-$y$ prefix of each selected record. We denote each resulting dataset by \texttt{SYN}$_{x.y}$. 

\paragraph{Comparison Measures.} We evaluated the {\em effectiveness} of the algorithms using: 
\begin{description}
\item[AvgRE]~An Average Relative Error measure, computed as  avg$_{i\in[1,1000]}\frac{k_i-k^*_i}{k^*_i}$, where $k_i^*$ is the size of the optimal solution produced by \textsc{BF}, and $k_i$ is the size of the solution produced by one of  the other tested algorithms. Both $k_i^*$ and $k_i$  are obtained by using, as query $q_i$, a record of the input dictionary selected uniformly at random. 
\item[AvgSS]~An Average Solution Size measure computed as avg$_{i\in[1,1000]}k^*_i$ for \textsc{BF} and avg$_{i\in[1,1000]}k_i$ for any other algorithm. 
\end{description}

We evaluated {\em efficiency} by reporting avg$_{i\in[1,1000]}t_i$, where $t_i$ is the elapsed time of a  tested algorithm to obtain a solution for query $q_i$ over the input dictionary. 

\paragraph{Execution Environment.} In our experiments we used a PC with Intel Xeon E5-2640@2.66GHz and 160GB RAM running GNU/Linux, and a \texttt{gcc} v.7.3.1 compiler at optimization level \texttt{-O3}. 

\begin{figure}[!t]\centering
    \begin{subfigure}[b]{0.32\textwidth}
       \includegraphics[scale=0.31]{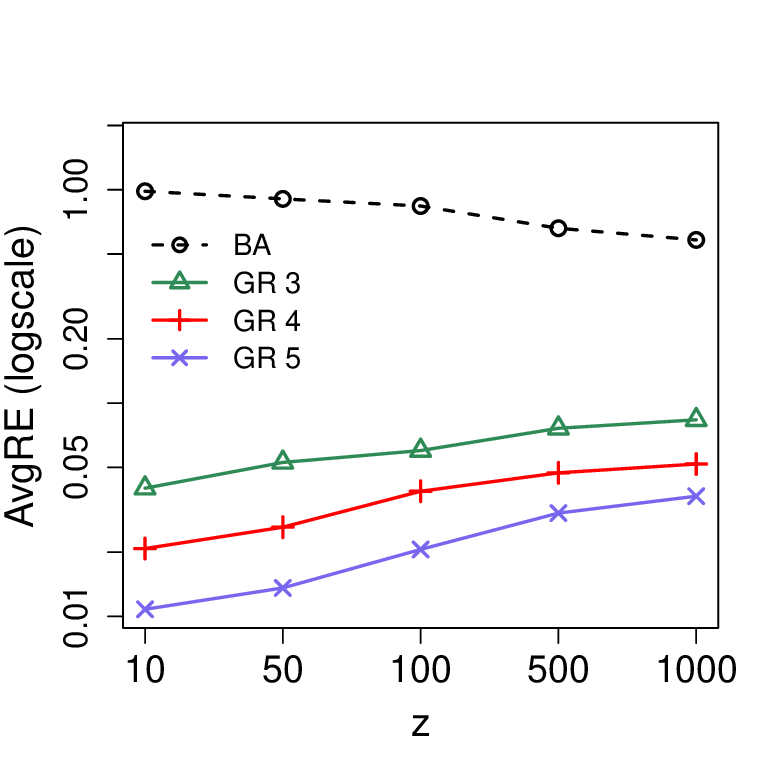}
        \caption{$d\!=\!952,\!864$ and $\ell=15$}
        \label{avgre1}
    \end{subfigure}
     \begin{subfigure}[b]{0.32\textwidth}
       \includegraphics[scale=0.31]{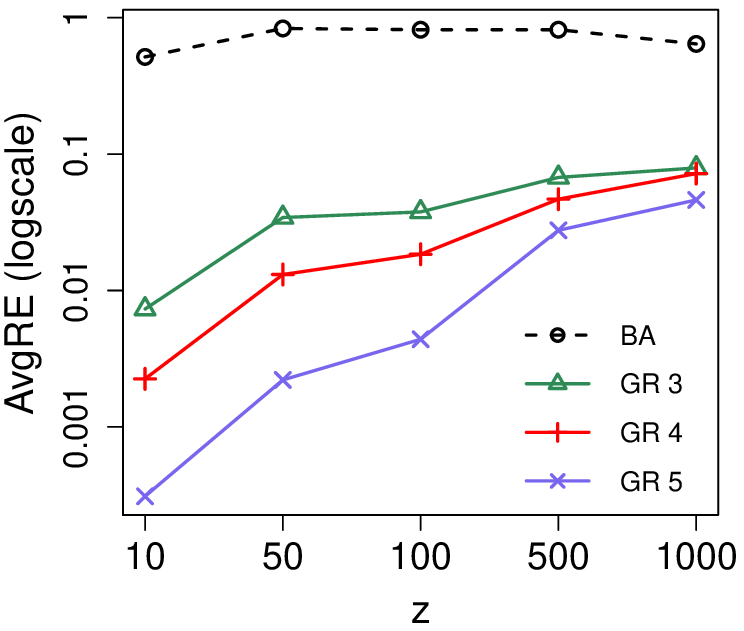}
        \caption{$d\!=\!342,\!472$ and $\ell=15$}
        \label{avgre2}
    \end{subfigure}
        \begin{subfigure}[b]{0.32\textwidth}
       \includegraphics[scale=0.27]{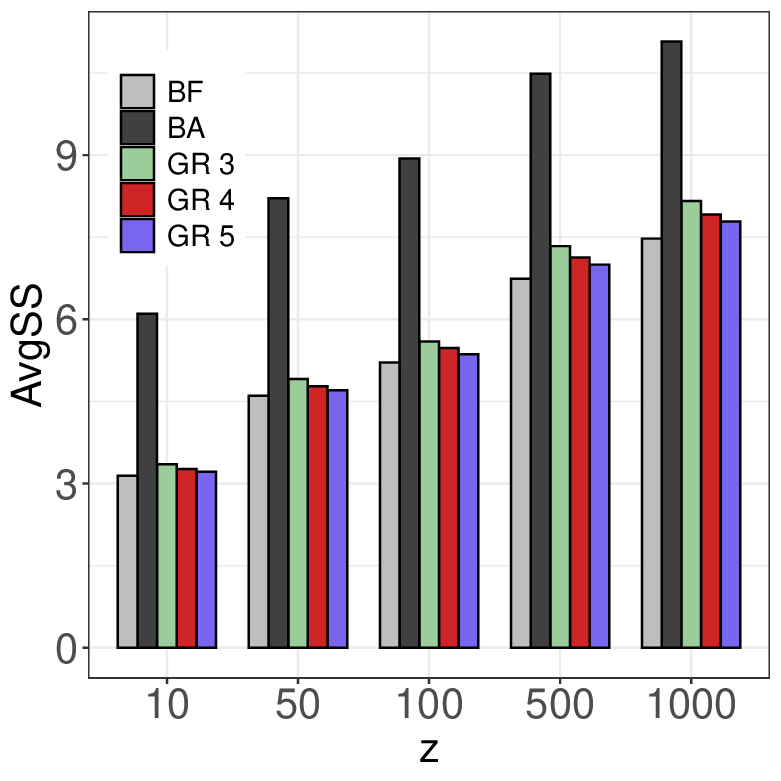}
        \caption{$d\!=\!952,\!864$ and $\ell=15$}
        \label{fullname_vs_z_ss}
    \end{subfigure}
    \caption{AvgRE (in logscale) vs.~$z$ computed for (a) \texttt{FS} and (b) \texttt{FCi}; (c) AvgSS vs.~$z$ for \texttt{FS}.} \label{avgre}
    
\end{figure}

\begin{figure}[!t]\centering
      \begin{subfigure}[b]{0.32\textwidth}
       \includegraphics[scale=0.27]{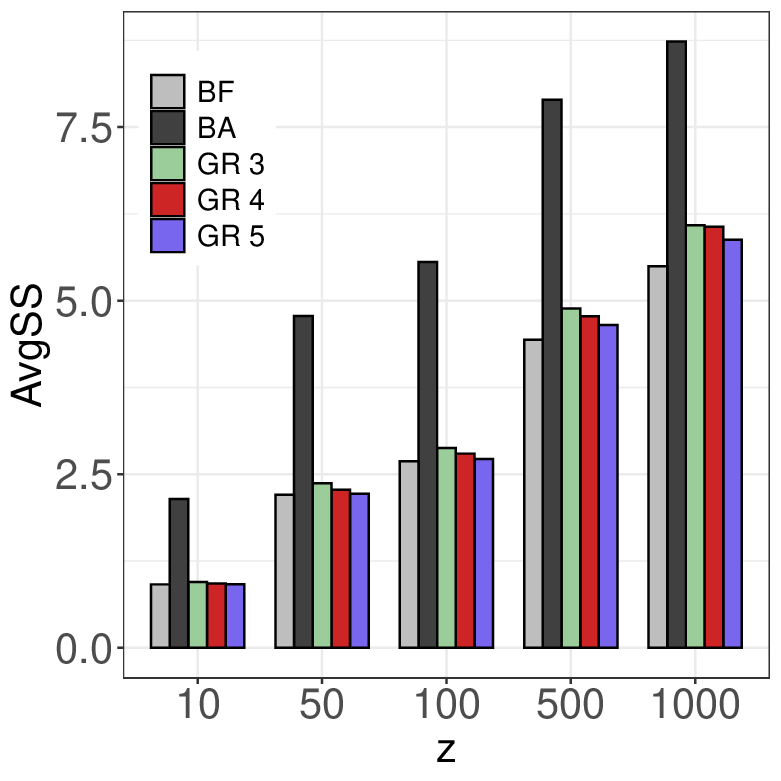}
        \caption{$d\!=\!342,\!472$ and $\ell=15$}
        \label{avgss1}
    \end{subfigure}
    \begin{subfigure}[b]{0.32\textwidth}
       \includegraphics[scale=0.27]{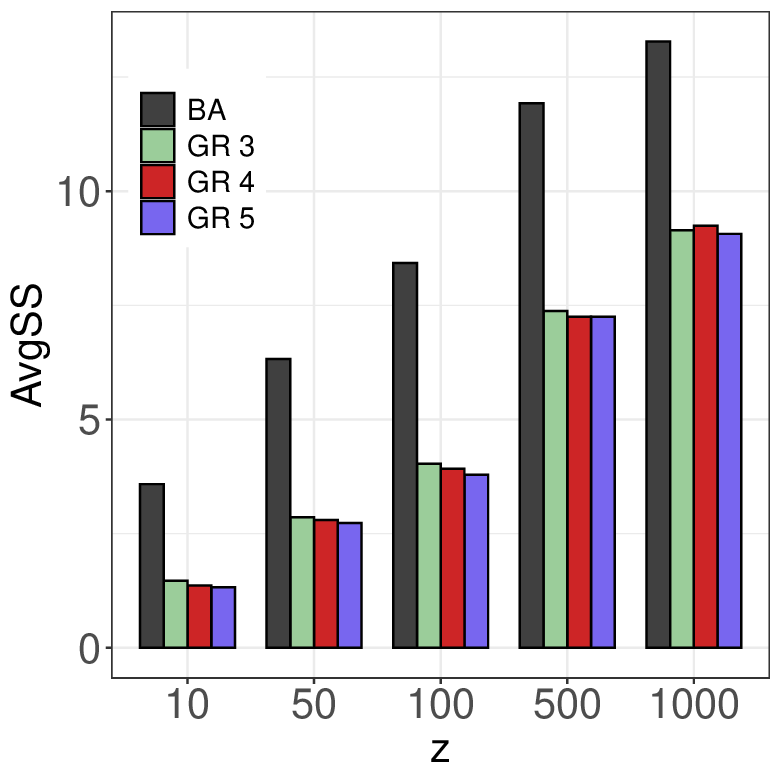}
        \caption{$d\!=\!342,\!472$ and $\ell=30$}
        \label{avgss2}
    \end{subfigure}
    \begin{subfigure}[b]{0.32\textwidth}
       \includegraphics[scale=0.27]{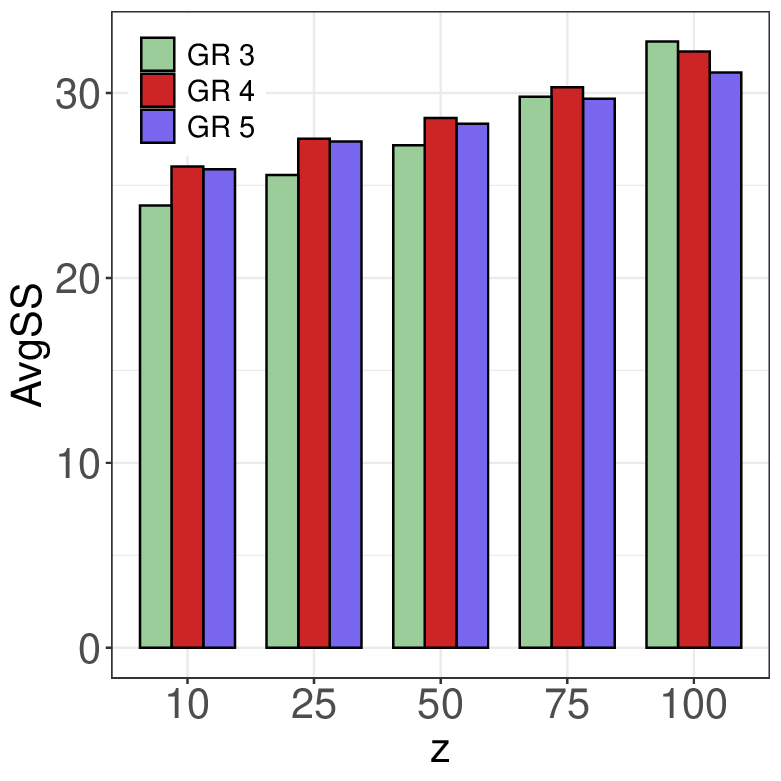}
         \caption{$d\!=\!8,\!238$ and $\ell=45$}
        \label{avgss3}
    \end{subfigure}
    
    \caption{AvgSS vs.~$z$ for (a) \texttt{FCi}. (b) \texttt{FCiCo} (\textsc{BF} did not produce results for any $z$ within 48 hours), and (c) \texttt{FSCiCo}  (\textsc{BF} and \textsc{BA} did not produce results for any $z$ within 48 hours. The results of \textsc{GR} for  $z>100$ are omitted because AvgSS~$>40$ which is close to $\ell=45$).} \label{avgss}
\end{figure}

\paragraph{Effectiveness.} Figures~\ref{avgre} and~\ref{avgss1} show that \textsc{GR} produced nearly-optimal solutions, significantly outperforming \textsc{BA}. In Figure~\ref{avgre1}, the solutions of \textsc{GR} 3 were no more than $9\%$ worse than the optimal, while those of \textsc{BA} were up to $95\%$ worse. In Figure~\ref{avgss1}, the average solution size of \textsc{BF} was $5.4$ vs.~$5.9$ and $9$, for the solution size of  \textsc{GR} 3 and \textsc{BA}, respectively. 

In Figures~\ref{avgss2} and~\ref{avgss3}, we examined the effectiveness of \textsc{GR} for larger $\ell$ values. Figure~\ref{avgss2} shows that the solution size of \textsc{GR} 3 was at least $31\%$ and up to $60\%$ smaller than that of \textsc{BA} on average, while Figure~\ref{avgss3} shows that the solution of \textsc{GR} 3 was comparable to that of \textsc{GR} 4 and 5.  We omit the results for \textsc{BF} from Figures~\ref{avgss2} and~\ref{avgss3}  and those for \textsc{BA} from Figure~\ref{avgss3}, as these algorithms did not produce results for all queries within 48  hours, for any $z$. This is because, unlike \textsc{GR}, \textsc{BF} does not scale well with $\ell$ and \textsc{BA} does not scale well with the solution size, as we will explain later. 

Note that increasing $\tau$ generally increases the effectiveness of \textsc{GR} as it computes more positions of wildcards per iteration. However, even with $\tau=3$, it remains competitive to \textsc{BF}.  

\paragraph{Efficiency.} Having shown that \textsc{GR} produced nearly-optimal solutions, we now show that it is comparable in terms of efficiency or faster than \textsc{BA} for synthetic data.  (\textsc{BF} was at least two orders of magnitude slower than the other methods on average and thus we omit its results.) The results for \texttt{NCVR} were qualitatively similar (omitted).  Figure~\ref{rta} shows that \textsc{GR} spent, on average, the same time for a query as \textsc{BA} did. However, it took significantly (up to $36$ times) less time than \textsc{BA} for queries with large solution size $k$. This can be seen from Figure~\ref{rttest}, which shows the time each query with solution size $k$ took; the results for \textsc{GR} 3 and 4 were similar and thus omitted. The reason is that \textsc{BA} updates the hypergraph every time a node is added into the solution, which is computationally expensive when $k$ is  large.  Figures~\ref{rtc} and~\ref{rtd} show that all algorithms scaled sublinearly with $d$ and with $z$, respectively. The increase with $d$ is explained by the time complexity of the methods. The slight increase with $z$ is because $k$  gets larger, on average, as $z$ increases (see Figure~\ref{avgss_vs_z} next, in which we also show the average solution size for the experiments in Figures~\ref{rta} and~\ref{rtc}). \textsc{GR} 3 and 4 performed similarly to each other, being faster than \textsc{GR} 5 in all experiments as expected: increasing $\tau$ from 3 or 4 to 5 trades-off effectiveness for efficiency.

\begin{figure}[!t]\hspace{-4mm}
    \begin{subfigure}[b]{0.25\textwidth}
       \includegraphics[trim={0 0 0 0},clip,scale=0.24]{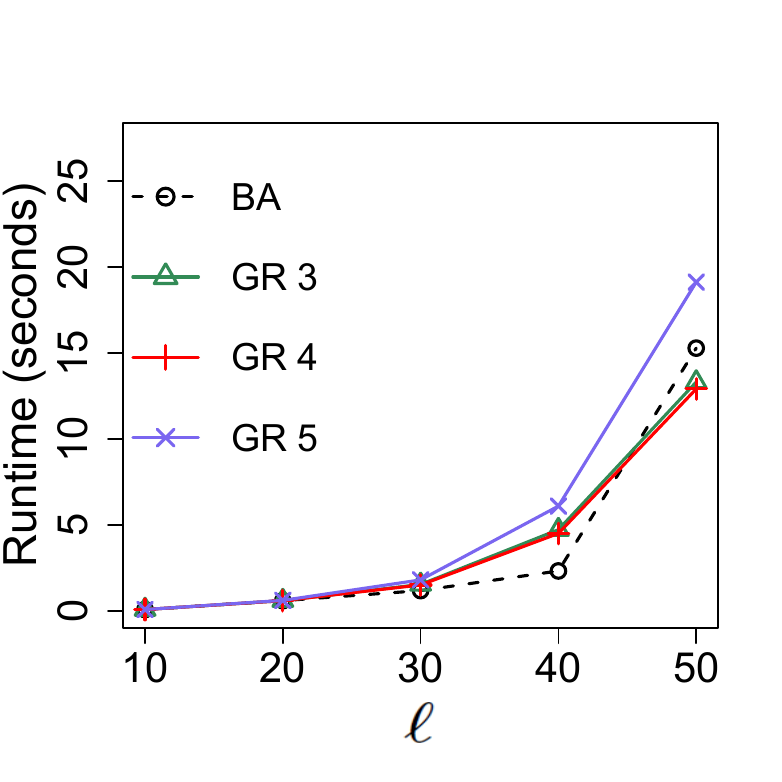}
        \caption{$d\!=\!6\!\cdot\! 10^6$ and $z\!=\!100$}
        \label{rta}
    \end{subfigure}
        \begin{subfigure}[b]{0.21\textwidth}
       \includegraphics[trim={0 0 0 0},clip,scale=0.28]{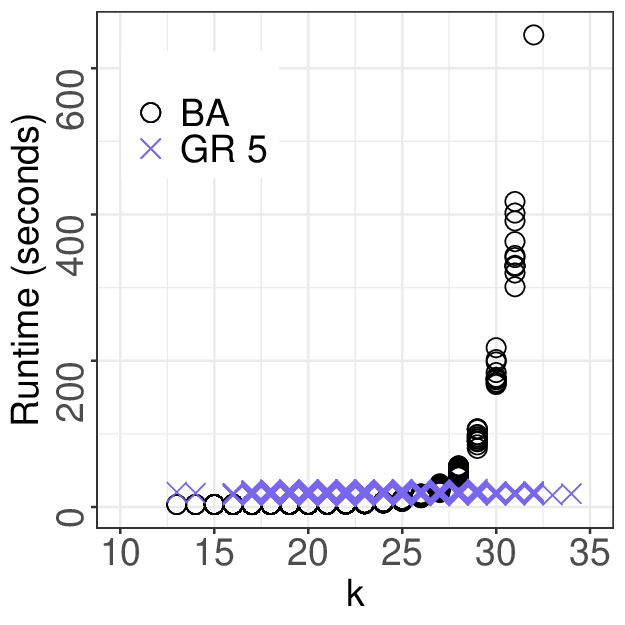}
        \caption{$\ell\!=\!50$ and $z\!=\!100$}
        \label{rttest}
    \end{subfigure}\hspace{+4mm}
    \begin{subfigure}[b]{0.24\textwidth}
       \includegraphics[scale=0.24]{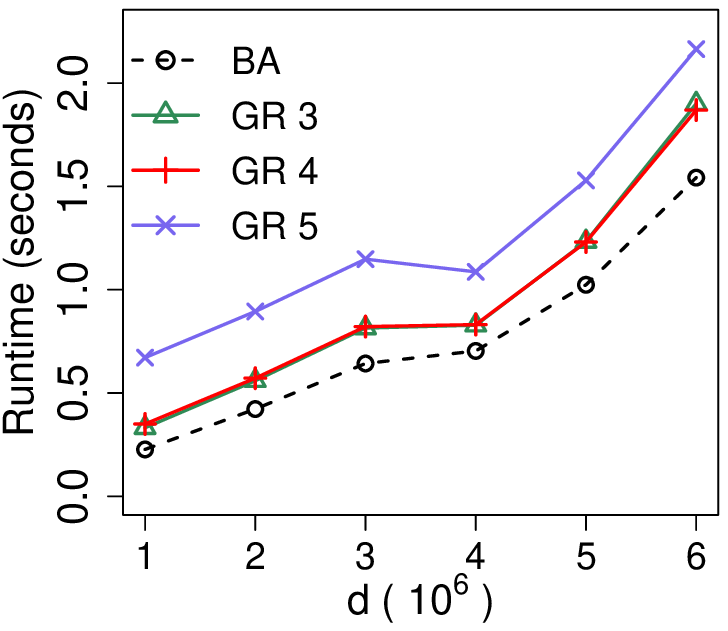}
        \caption{$\ell\!=\!30$ and $z\!=\!100$}
        \label{rtc}
    \end{subfigure}\hspace{+1mm}
    \begin{subfigure}[b]{0.25\textwidth}
       \includegraphics[trim={0 0 0 0},clip,scale=0.24]{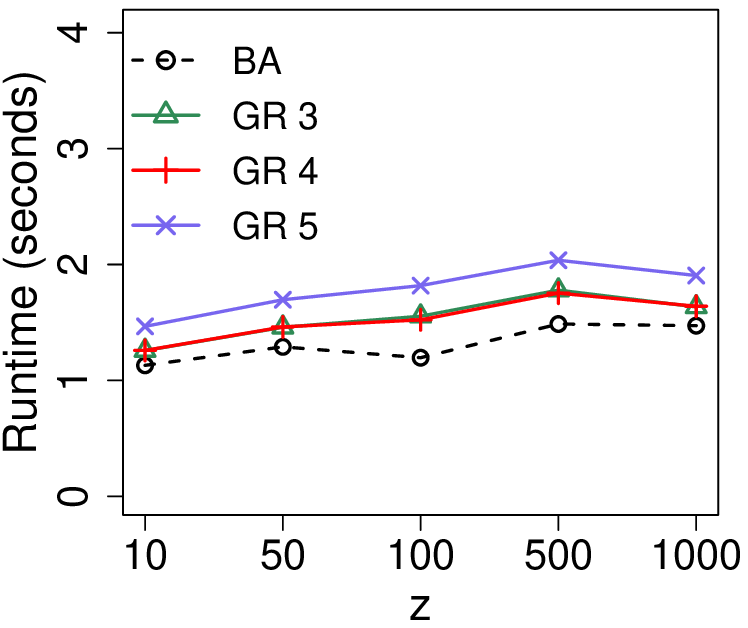}
        \caption{$d\!=\!6\!\cdot\!10^6$ and $\ell\!=\!30$}
        \label{rtd}
    \end{subfigure}
    ~ 
    
    \caption{Efficiency vs.~(a) $\ell$ for $\texttt{SYN}_{6.\ell}$, (b) $k$ for \texttt{SYN}, (c) $d$ for \texttt{SYN}$_{x.30}$, $x\in\{1\cdot 10^6, 2\cdot 10^6, \ldots, 6\cdot 10^6\}$, and (d) $z$ for $\texttt{SYN}_{6.30}$. The results of \textsc{BF} are omitted, because it was slower than other methods by at least two orders of magnitude on average.}  \label{runtime}
\end{figure}

\paragraph{Average Solution Size.} Figures \ref{avgss_vs_l}, \ref{avgss_vs_d}, and \ref{avgss_vs_z} show the average solution size in the experiments of Figure \ref{rta}, \ref{rtc}, and \ref{rtd}, respectively. 
Observe that the results are analogous to those obtained using the \texttt{NCVR} datasets: \textsc{GR} outperforms \textsc{BA} significantly. Also, observe in Figure \ref{avgss_vs_z} that the solution size of each tested algorithm gets larger, on average, as $z$ increases. 

\begin{figure}[!t]\centering
      \begin{subfigure}[b]{0.32\textwidth}
       \includegraphics[scale=0.29]{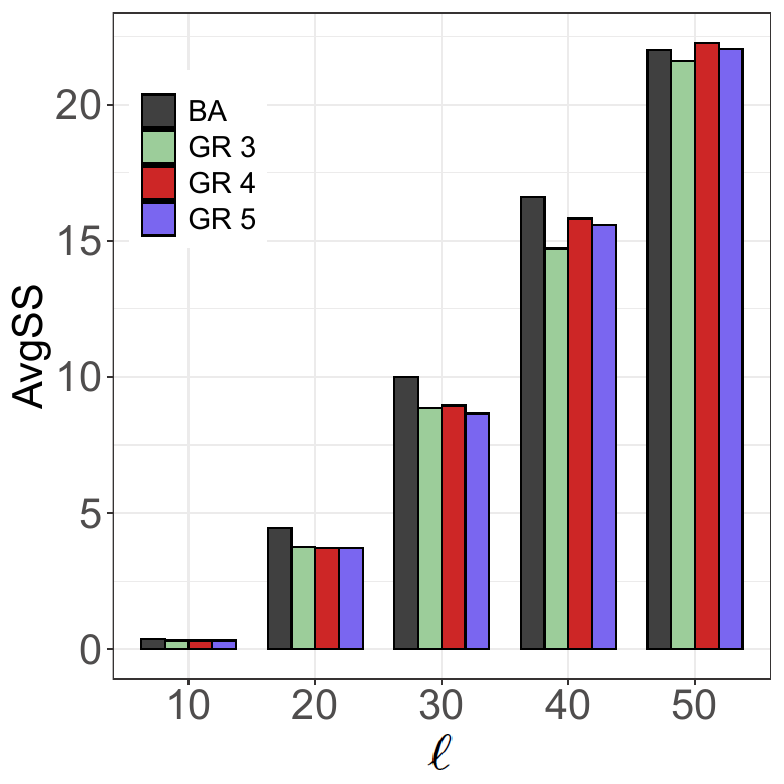}
        \caption{$d=6\cdot 10^6$, $z=100$}
        \label{avgss_vs_l}
    \end{subfigure}
    \begin{subfigure}[b]{0.32\textwidth}
       \includegraphics[scale=0.29]{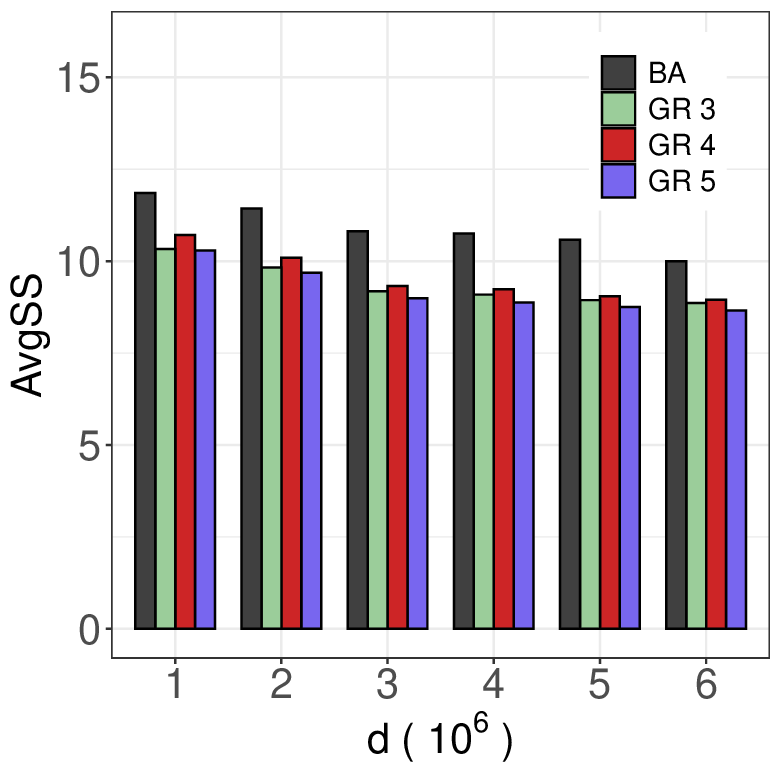}
        \caption{ $\ell=30$, $z=100$}
        \label{avgss_vs_d}
    \end{subfigure}
    \begin{subfigure}[b]{0.32\textwidth}
       \includegraphics[scale=0.29]{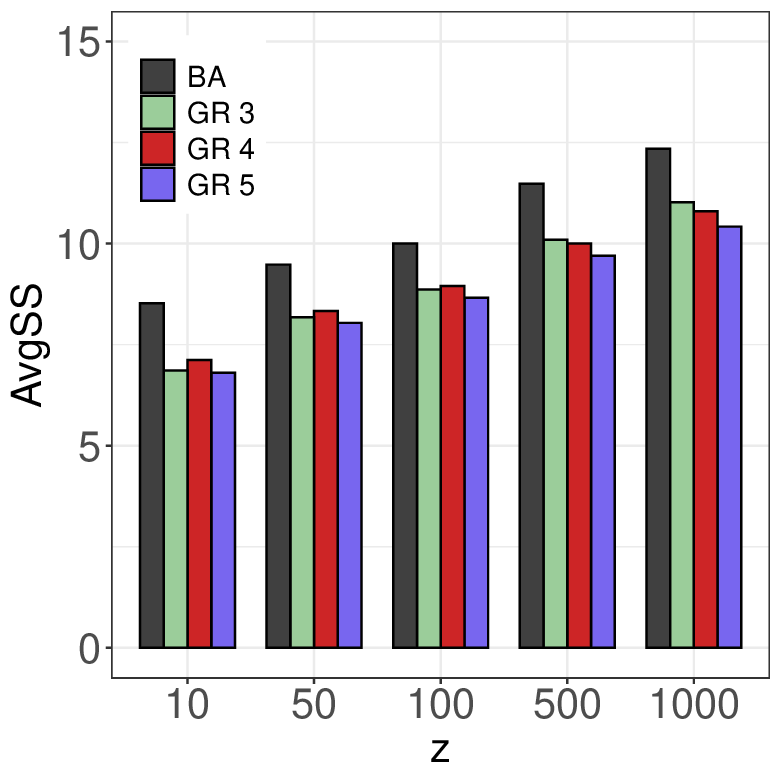}
        \caption{$d=6\cdot 10^6$, $\ell=30$}
        \label{avgss_vs_z}
    \end{subfigure}
    
    \caption{AvgSS vs.~(a) $\ell$ for $\texttt{SYN}_{6.\ell}$, (b) $d$ for \texttt{SYN}$_{x.30}$, $x\in \{1\cdot 10^6, 2\cdot 10^6, \ldots, 6\cdot 10^6\}$, and (c) $z$ for $\texttt{SYN}_{6.30}$.} \label{fig:avgss_appendix}
\end{figure}

\paragraph{Summary.} We have presented an extensive experimental evaluation demonstrating the effectiveness and efficiency of the proposed heuristic on real-world datasets used in record linkage, as well as on synthetic datasets. In the experiments that we have performed, the proposed heuristic: (I) found nearly-optimal solutions for varying values of $z$ and $\ell$, even when applied with a small $\tau$; and (II) scaled as predicted by the complexity analysis, requiring fewer than $3$ seconds for $d=6\cdot10^6$ in all tested cases. Our experimental results suggest that our methods can inspire solutions in large-scale real-world systems, where no sophisticated algorithms for \PPSM are being used. 

\section{Open Questions}\label{sec:finale}
The following questions of theoretical nature remain unanswered:

\begin{enumerate}
    \item Can we improve on the exact $\cO(d\ell+(d \ell)^{k/3})$-time algorithm presented in~\cref{sec:exact} for \PPSM and $k=\cO(1)$ using fast matrix multiplication~\cite{DBLP:conf/soda/AlmanW21,DBLP:conf/issac/Gall14a} or show that algebraic techniques cannot help, e.g., via the \textsc{$(c,k)$-Hyperclique} problem?
    \item Can we improve on the $\cO(2^{\ell/2}(2^{\ell/2}+\tau)\ell)$-time and  $\cO(2^{\ell}d^2/\tau^2+2^{\ell/2}d)$-space trade-off presented in~\cref{sec:DS}
    for the data structure answering $q,z$ \PPSM queries?
\end{enumerate}

\bibliographystyle{plainurl}
\bibliography{references}

\end{document}